\documentclass[11pt]{article}
\usepackage[utf8]{inputenc}
\usepackage{stackrel}
\usepackage{color}

\usepackage{aliascnt}
\usepackage{fullpage}
\usepackage{color}
\usepackage{comment}

\usepackage{bbm}
\usepackage{algorithm}
\usepackage{algorithmic} 
\usepackage{microtype}  
\usepackage{array}
\usepackage{multirow}
\usepackage{amsmath} 
\usepackage{amssymb}
\usepackage{amsthm} 
\usepackage{amsfonts}
\usepackage{thmtools}

\usepackage{xcolor}
\usepackage{nameref}
\definecolor{ForestGreen}{rgb}{0.1333,0.5451,0.1333}
\definecolor{DarkRed}{rgb}{0.65,0,0}
\definecolor{Red}{rgb}{1,0,0}
\usepackage[linktocpage=true,
pagebackref=true,colorlinks,
linkcolor=DarkRed,citecolor=ForestGreen,
bookmarks,bookmarksopen,bookmarksnumbered]
{hyperref}
\usepackage{cleveref}

\declaretheorem[numberwithin=section]{theorem}
\declaretheorem[numberlike=theorem]{lemma}
\declaretheorem[numberlike=theorem,name=Lemma]{lem}
\declaretheorem[numberlike=theorem]{fact}

\declaretheorem[numberlike=theorem,name=Proposition]{prop}

\declaretheorem[numberlike=theorem,name=Corollary]{cor}
\declaretheorem[numberlike=theorem]{claim}

\declaretheorem[numberlike=theorem,,name=Remark,style=remark]{remark}
\crefname{remark}{Remark}{Remarks}

\declaretheorem[numberlike=theorem,style=definition]{definition}
\declaretheorem[numberlike=theorem,name=Definition,style=definition]{defn}
\declaretheorem[numberlike=theorem,name=Open Question]{question}
\renewcommand{\cref}[1]{\Cref{#1}}

\def\draft{1} %
\newcommand{\calD}{\mathcal{D}}
\newcommand{\N}{\mathbb{N}}
\newcommand{\NN}{\mathbb{N}}
\newcommand{\AAA}{\mathcal{A}}
\newcommand{\BBB}{\mathcal{B}}
\newcommand{\SSS}{\mathcal{S}}
\newcommand{\DDD}{\mathcal{D}}
\newcommand{\MMM}{\mathcal{M}}
\newcommand{\WWW}{\mathcal{W}}

\newcommand{\set}[1]{\{ #1 \}}
\newcommand{\E}{\mathbb{E}}

\newcommand{\remove}[1]{}

\ifnum\draft=1
 \newcommand{\authnote}[2]{{\bf [{\color{red} #1's Note:} {\color{blue} #2}]}}
\else
 \newcommand{\authnote}[2]{}
\fi

\newcommand{\polylog}{{\rm polylog}}

\global\long\def\A{{\cal A}}

\global\long\def\Otil{\tilde{O}}
\global\long\def\poly{\mathrm{poly}}
\global\long\def\polylog{\mathrm{polylog}}
\global\long\def\ibegin{1}
\global\long\def\ifin{T}

\global\long\def\primed{\mathtt{pMedian}}
\global\long\def\med{\mathrm{med}}

\global\long\def\total{\mathrm{total}}
\global\long\def\Ex{\mathbb{E}}
\global\long\def\dist{\mathrm{dist}}
\global\long\def\adversary{\texttt{Adversary}}
\global\long\def\output{\mathtt{out}}
\global\long\def\outtil{\widetilde{\mathtt{out}}}
\global\long\def\fail{\mathrm{fail}}
\global\long\def\xvec{\vec{x}}
\global\long\def\acc{\mathrm{acc}}
\global\long\def\one{\boldsymbol{1}}
\global\long\def\paramc{200\cdot6s\epsilon_{\med}\cdot\sqrt{2T\ln(100/\beta)}}
\global\long\def\sparsify{\textsc{Sparsify}}
\global\long\def\vol{\mathrm{vol}}
\global\long\def\counter{\mathtt{cnt}}
\global\long\def\Gen{\mathtt{Gen}}
\global\long\def\Enc{\mathtt{Enc}}
\global\long\def\Dec{\mathtt{Dec}}
\global\long\def\Oracle{\mathrm{Oracle}}
\newcommand{\accgame}{\mathsf{Acc}}
\newcommand{\RO}{\mathcal{R}}

\title{Dynamic Algorithms Against an Adaptive Adversary:\\ Generic Constructions and Lower Bounds}

\author{Amos Beimel%
\thanks{Department of Computer Science, Ben-Gurion University. {\tt amos.beimel@gmail.com}. Work partially funded by ERC grant 742754 (project NTSC), by the Israel Science Foundation grant no.~391/21, and by the Cyber Security Research Center at Ben-Gurion University}
\and
Haim Kaplan\thanks{Tel Aviv University and Google Research. {\tt haimk@tau.ac.il}. Partially supported by Israel Science Foundation (grant 1595/19),  and the Blavatnik Family Foundation.}
\and
Yishay Mansour\thanks{Tel Aviv University and Google Research. {\tt mansour.yishay@gmail.com}. Work partially funded from the European Research Council (ERC) under the European Union’s Horizon 2020 research and innovation program (grant agreement No. 882396), by the Israel Science Foundation (grant number 993/17), Tel Aviv University Center for AI and Data Science (TAD), and the Yandex Initiative for Machine Learning at Tel Aviv University.}
\and
Kobbi Nissim\thanks{Department of Computer Science, Georgetown University. {\tt kobbi.nissim@georgetown.edu}. Work partially funded by NSF grant No.~2001041 and by a gift to Georgetown University.}
\and
Thatchaphol Saranurak\thanks{University of Michigan. {\tt thsa@umich.edu}.}
\and
Uri Stemmer\thanks{Tel Aviv University and Google Research. {\tt u@uri.co.il}. Partially supported by the Israel Science Foundation (grant 1871/19) and by Len Blavatnik and the Blavatnik Family foundation.}
}
\date{November 7, 2021}

\begin{document}

\maketitle

\begin{abstract}
A dynamic algorithm against an adaptive adversary is required to be correct when the adversary chooses the next update after seeing the previous outputs of the algorithm. We obtain faster dynamic algorithms against an adaptive adversary and separation results between what is achievable in the oblivious vs. adaptive settings. To get these results we exploit techniques from differential privacy, cryptography, and adaptive data analysis.
    
We give a general reduction transforming a dynamic algorithm against an oblivious adversary to a dynamic algorithm robust against an adaptive adversary. This reduction maintains several copies of the oblivious algorithm and uses differential privacy to protect their random bits. Using this reduction we obtain dynamic algorithms against an adaptive adversary with improved update and query times for global minimum cut, all pairs distances, and all pairs effective resistance.
 
We further improve our update and query times by showing how to maintain a sparsifier over an expander decomposition that can be refreshed fast. This fast refresh enables it to be robust against what we call a blinking adversary that can observe the output of the algorithm only following refreshes. We believe that these techniques  will prove useful for additional problems.
 
 On the flip side, we specify dynamic problems that, assuming a random oracle, every dynamic algorithm that solves them against an adaptive adversary must be polynomially slower than a rather straightforward dynamic algorithm that solves them against an oblivious adversary. We first show a separation result for a search problem and then show a separation result for an estimation problem. In the latter case our separation result draws from lower bounds in adaptive data analysis.
\end{abstract}

\section{Introduction}

Randomized algorithms are often analyzed under the assumption that their internal randomness is independent of their inputs. This is a reasonable assumption for {\em offline} algorithms, which get all their inputs at once, process it, and spit out the results. However, in {\em online} or {\em interactive} settings, this assumption is not always reasonable. For example, consider a {\em dynamic} setting where the input comes in gradually (e.g., a graph undergoing a sequence of edge updates), and the algorithm continuously reports some value of interest (e.g., the size of the minimal cut in the current graph). In such a dynamic setting, it might be the case that future inputs to the algorithm depend on its previous outputs, and hence, depend on its internal randomness. For example, consider a large system in which a dynamic algorithm is used to analyze data coming from one part of the system while answering queries generated by another part of the system, but these (supposedly) different parts of the system are connected via a feedback loop. In such a case, it is no longer true that the inputs of the algorithm are independent of its internal randomness.

Nevertheless, classical algorithms, even for interactive settings, are typically analyzed under the (not always reasonable) assumption that their inputs are independent of their internal randomness. (The reason is that taking these dependencies into account often makes the analysis significantly more challenging.) One approach for avoiding the problem is to make the system deterministic. This, however, is very limiting as randomness is essential for good performance in streaming algorithms, online algorithms, and dynamic algorithms---basically in every algorithmic area in which the algorithm runs interactively. This calls for the design of algorithms providing (provable) utility guarantees even when their inputs are chosen adaptively. Indeed, this motivated an exciting line of work, spanning different communities in theoretical computer science, all focused on this question. This includes works from the streaming community \cite{ben2020framework,HassidimKMMS20,woodruff2020tight,alon2021adversarial,KaplanMNS21,KaplanMNS21,Braverman2021adversarial,ACSS21,BeEO21,chakrabarti2021adversarially}, learning community \cite{HardtU14,dwork2015preserving,bassily2021algorithmic,steinke2017tight,gupta2021adaptive,chakrabarti2021adversarially}, and dynamic algorithms \cite{holm2001poly,NanongkaiS17,Wulff-Nilsen17,NanongkaiSW17,chuzhoy2020deterministic,BernsteinC16,Bernstein17,BernsteinChechikSparse,ChuzhoyK19,ChuzhoyS20,gutenberg2020decremental,gutenberg2020deterministic,GutenbergWW20,Chuzhoy21,BhattacharyaHI15,BhattacharyaHN16,BhattacharyaHN17,BhattacharyaK19,Wajc19,BhattacharyaK21deterministic}. 
We continue the study of this question for dynamic algorithms.

\medskip
Before presenting our new results, we make our setting more precise. Given an input $x$ that undergoes a sequence of updates, our goal is to maintain a valid solution to some predefined problem $P$ at any point in time. We consider both {\em estimation} and {\em search} problems. In an {\em estimation} problem the goal is to provide a $(1\pm \epsilon)$ approximation of a numeric quantity that is a function of the current input. An example is the global min-cut problem, where an update inserts or deletes an edge and the goal is to output a $(1\pm \epsilon)$ approximation of the size of the global min-cut. In a {\em search} problem, the response to the query is a non-numeric value. For example, in the search version of the global min-cut problem the goal is to output a cut whose size is not much larger than the size of the smallest cut.

Formally, given a problem $P$ (over a domain $X$) and an initial input $x_0\in X$, we consider a sequence of $m$ input updates $u_1,u_2,\dots,u_m$, where every $u_i$ is a function $u_i:X\rightarrow X$. After every such update $u_i$, the (current) input is replaced with $x_i\leftarrow u_i(x_{i-1})$, and our goal is to respond with a valid solution for $P(x_i)$.
We refer to the case where the sequence of input updates is fixed in advance as the {\em oblivious} setting. In this work, we focus on the {\em adaptive} setting, where the sequence of input updates may be chosen adaptively. We think of the entity that generates the inputs as an ``adversary'' whose goal is to force the algorithm to misbehave (either to err or to have a large runtime). Specifically, the adaptive setting is modeled by a two-player game between a (randomized) \texttt{Algorithm} and an \texttt{Adversary}. At the beginning, we fix a problem $P$, and the \texttt{Adversary} chooses the initial input $x_0$. Then the game proceeds in rounds, where in the $i$th round:

\begin{enumerate}
    \item The \texttt{Adversary} chooses an update $u_i$, thereby modifying the current input to $x_i\leftarrow u_i(x_{i-1})$. Note that $u_i$ may depend on all previous updates and outputs of the \texttt{Algorithm}.
    \item The \texttt{Algorithm} processes the new update $u_i$ and outputs its current response $z_i$.
\end{enumerate}

\begin{remark}
For simplicity, in the above two-player game we focused on the case where the algorithm outputs a response after every update. Actually, in later sections of the paper, we will make the distinction between an {\em update} and a {\em query}. Specifically, 
in every time step the \texttt{Adversary} poses either an {\em update} (after which the \texttt{Algorithm} updates its data structure and outputs nothing) or a {\em query} (after which the \texttt{Algorithm} outputs a response). Moreover, in some of our results, we will separate between the preprocessing time (the period of time after the algorithm obtains its initial input $x_0$, and before it obtains its first update) and the update time and query time. 
\end{remark}

We say that the \texttt{Algorithm} solves $P$ in the adaptive setting with amortized update (and query) time $t$ if for any \texttt{Adversary}, with high probability, in the above two-player game,
\begin{enumerate}
    \item For every $i$ we have that $z_i$ is a valid solution for $P(x_i)$.
    \item The \texttt{Algorithm} runs in amortized time $t$ per update.
\end{enumerate}

\paragraph{Notation.} We refer to algorithms in the adaptive setting as {\em algorithm that work against an adaptive adversary}, or {\em adaptive algorithms} in short. Analogously, we refer to algorithms in the oblivious setting as {\em algorithm that work against an oblivious adversary}, or {\em oblivious algorithms} in short.

\subsection{Our Contributions}

In this paper we take a general perspective into studying dynamic algorithms against an adaptive adversary. 
On the positive side, we develop a general technique for transforming an oblivious algorithm into an adaptive algorithm. 
We show that our general technique results in new adaptive algorithms for several graph problems of interest.
On the negative side, we prove separation results. Specifically, we present dynamic problems that can be trivially  solved against oblivious adversaries, but, under certain assumptions, require a significantly higher computation time when the adversary is adaptive. This is the first separation between the oblivious setting and the adversarial setting for dynamic algorithms. In particular, this is the first separation between randomized and
deterministic dynamic algorithms, since deterministic algorithms always work against an adaptive adversary.

\subsubsection{Our positive results}

We describe a generic black-box reduction
to obtain a dynamic algorithm against an adaptive adversary from an oblivious one.
This reduction can be applied to any oblivious dynamic algorithm for an estimation problem. 
We then apply our generic reduction to obtain adaptive algorithms for several well-studied graph problems. 
To speed up the adaptive algorithms that we get from our generic reduction, we construct a sparsifier over a dynamic expander decomposition with fast initialization time. By combining our reduction technique with this sparsifier
we substantially improve the amortized update time. 
We obtain the following theorem.

\begin{theorem}
\label{thm:app}Given a weighted graph $G$ with $n$ vertices undergoing
edge insertions and deletions, there are dynamic algorithms against
an adaptive adversary for the following problems:
\begin{itemize}
\item \textbf{(Global min cuts):} $(1+\epsilon)$-approximation using $\Otil(n^{3/4})$
amortized update and query time (see \Cref{cor:global mincut sparse}).
\item \textbf{(All-pairs effective resistance):} $(1+\epsilon)$-approximation
using $n^{3/4+o(1)}$ amortized update and query time (see \Cref{cor:global mincut sparse}).%
\item \textbf{(All-pairs distances): }$\log^{3i+O(1)}n$-approximation using
$n^{1/2+1/(2i)+o(1)}$ amortized update and query time, for any integer
$i$ (see \Cref{cor:distance sparse}).
\item \textbf{(All-pairs distances with better approximation):} \textup{$\log n\cdot\poly(\log\log n)$-approximation
using} $\Otil(m^{4/5})$ amortized update and query time, where $m$
is the current number of edges (see \Cref{cor:distance accurate slow}).
\end{itemize}
\end{theorem}

\begin{table}
\begin{tabular}{|>{\raggedright}m{0.19\textwidth}|>{\centering}m{0.13\textwidth}|>{\centering}m{0.15\textwidth}|>{\centering}m{0.24\textwidth}|>{\centering}m{0.17\textwidth}|}
\hline 
Problems & Approx.  & Previous 

oblivious algo. & Previous 

adaptive algo. & Our new 

adaptive algo.\tabularnewline
\hline 
\hline 
Global min cut & $(1+\epsilon)$ & $n^{1/2}$ 

\cite{thorup2007fully} & $m$ 

(use static algo.) & $n^{3/4}$\tabularnewline
\hline 
All-pairs effective resistance & $(1+\epsilon)$ & $n^{2/3+o(1)}$ 

\cite{chen2020fast} & $m$ 

(use static algo.) & $n^{3/4+o(1)}$\tabularnewline
\hline 
\multirow{2}{0.19\textwidth}{All-pairs distances } & $\log^{3i+O(1)}n$ & $n^{1/i+o(1)}$ 

\cite{forsterGH21} & $n$ 

(implicit in \cite{bernstein2020fully}) & $n^{1/2+1/(2i)+o(1)}$\tabularnewline
\cline{2-5} 
 & $\log n\cdot$

$\poly\log\log n$ & $n^{2/3+o(1)}$ 

\cite{chen2020fast} & $m$ 

(use static algo.) & $m^{4/5}$ 
\tabularnewline
\hline 
\end{tabular}

\caption{\label{tab:compare}A comparison between previous oblivious and adaptive
algorithms, and our new adaptive algorithm. The bound shown is the
maximum of update time and query time. We omit $\protect\polylog(n)$
factors. }
\end{table}

\Cref{thm:app} gives the first adaptive algorithms with both update
and query time of $o(n)$ for all the problems from \Cref{thm:app},
except the result of the $\log n\cdot\poly(\log\log n)$-approximate
all-pairs distances problem  where our algorithm
gives the first $o(m)$ bound. 
In \Cref{tab:compare}, we compare our results to %
previously known results. For all problems we consider, essentially
the only known technique against an adaptive adversary is to maintain
a sparsifier of the graph \cite{bernstein2020fully} and simply query on top of the sparsifier.

\subsubsection{Our negative results}

We present two separation results, one for a search problem and one for an estimation problem. 
Our separation results rely on (unproven) computational assumptions, which hold in the random oracle model.\footnote{A random oracle is an infinite random string $\RO$ such that the algorithm can read the $i$-th bit in $\RO$, $\RO[i]$, at the cost of one time unit. We assume than each bit of $\RO$ is uniformly distributed and independent of all other bits of $\RO$, thus, the only way to get any information on $\RO[i]$ is to read this bit.} Assuming a random oracle is a common assumption in cryptography, starting in the seminal work of Bellare and Rogaway~\cite{BellareR93}.
Many practical constructions are first designed and proved assuming a random oracle and then implemented using a cryptographic hash function replacing the random oracle. 
This is known as the random oracle methodology.\footnote{
We remark that the random oracle methodology is a heuristic. Canetti, Halevi, and Goldreich provide constructions of cryptographic schemes specifically tailored such that they become insecure under any instantiation of the random oracle with a computable function~\cite{CanettiGH98}.
} Thus, heuristically, the computational assumptions we make hold for cryptographic hash functions. We obtain the following theorem.

\begin{theorem}[informal]\label{thm:negative_intro}
Under some computational assumptions (or, alternatively, in the random oracle model):
\begin{enumerate}
    \item For any constant $c>1$, there is a dynamic search problem $P_{\rm search}^n$, where $n$ is a parameter controlling the instance size, that can be solved in the oblivious setting with amortized update time $O(n)$, but requires amortized update time $\Omega(n^c)$ in the adversarial setting, even if the algorithm is allowed time $2^{0.5n}$ in the preprocessing stage. 
   
    \item There is a dynamic estimation problem $P_{\rm est}^n$, where $n$ is a parameter controlling the instance size, 
    that can be solved in the oblivious setting with total time $O(n^6)$ over $O(n^2)$ updates, but requires total time $\Omega(n^7)$ over $O(n^2)$ updates in the adversarial setting. 
\end{enumerate}
\end{theorem}

We note that our lower bound for the estimation problem $P_{\rm est}^n$ matches what we would get by applying our positive result (our generic reduction) to the oblivious algorithm that solves $P_{\rm est}^n$.

\subsection{Technical Overview}
We give a technical overview of our results. Any informalities made herein will be removed in the sections that follow. 

\subsubsection{Generic transformation using differential privacy}
Differential privacy \cite{dwork2006calibrating} is a mathematical definition for privacy that aims to enable statistical analyses of datasets while providing strong guarantees that individual level information does not leak. Over the last few years, differential privacy has proven itself to be an important algorithmic notion (even when data privacy is not of concern), and has found itself useful in many other fields, such as machine learning, mechanism design, secure computation, probability theory, secure storage, and more \cite{McSherryT07,bassily2021algorithmic,NissimS19,KellarisKNO17,BeimelHMO18}.

Recall that the difficulty in the adaptive setting arises from potential dependencies between the inputs of the algorithm and its internal randomness. 
Our transformation uses a technique, introduced by \cite{HassidimKMMS20} (in the context of streaming algorithms), for using {\em differential privacy} to protect not the input data, but rather the internal randomness of algorithm. As \cite{HassidimKMMS20} showed, this can be used to limit (in a precise way) the dependencies between the internal randomness of the algorithm and its inputs, thereby allowing to argue more easily about the utility guarantees of the algorithm in the adaptive setting. Following \cite{HassidimKMMS20}, this technique was also used by \cite{ACSS21,BeEO21} for streaming algorithms and by \cite{gupta2021adaptive} for machine unlearning. We adapt this technique and connect it to the setting of dynamic algorithms in general and dynamic algorithms for graph problems in particular.

Informally, our generic transformation can be described as follows.

\begin{enumerate}
    \item Let $T$ be a parameter, and let $\A$ be an oblivious algorithm. 
    \item Before the first update arrives, we initialize $c=\Otil(\sqrt{T})$
independent copies of $\A$, with the initial input $x_{0}$.
\item For $T$ time step $i=1,2,3,\dots,T$:
\begin{enumerate}
    \item Obtain an update $u_{i}$.
    \item Feed the update $u_{i}$ to {\em all} of the copies of $\A$.
    \item Sample $\Otil(1)$ of the copies of $\A$, query them, aggregate their responses with differential privacy, and output the aggregated value.
\end{enumerate}
\item Reset all of the copies of $\A$, i.e., re-initialize each copy on the current input with fresh randomness, and goto step 3. 
\end{enumerate}

It can be shown that this construction satisfies differential privacy w.r.t.\ the internal randomnesses of each of the copies of algorithm $\AAA$. The intuition is that by instantiating $c=\Otil(\sqrt{T})$ copies of $\A$ we have enough ``privacy bidget'' to aggregate $T$ values privately (this follows from advanced composition theorems for differential privacy \cite{DworkRV10}). After exhausting our privacy budget, we reset all our data structures, and hence, reset our privacy budget. The total time we need for $T$ steps is 
$$
\tilde{O}\left( \sqrt{T}\cdot t_{\rm total} + T\cdot t_q  \right),
$$
where $t_{\rm total}$ is the total time needed to conduct $T$ updates to algorithm $\A$, and $t_q$ is the query time of algorithm $\A$. Instantiating this generic construction for the graph problems we study already gives new (and non trivial) results, but does not yet obtain the results stated in Theorem~\ref{thm:app} and in Table~\ref{tab:compare}. Specifically, this obtains all the results in \Cref{thm:app} except that the update and query time bounds depend on $m$, the number of edges, rather than $n$, the number of  vertices.

\subsubsection{The blinking adversary model}
\label{sec:blinking}

We observe that for the graph problems we are interested in, we can improve over the above generic construction as follows. We design oblivious algorithms with an extra property that allows us to ``refresh'' them faster than the time needed to initialize them from scratch. 
Formally, the ``refresh'' properties that we need are:
\begin{enumerate}
    \item[(i)] Hitting  the refresh button is {\em faster} than instantiating the algorithm from scratch.
    \item[(ii)] The algorithm maintains utility against a ``semi-adaptive'' adversary (which we call a {\em blinking adversary}), defined as follows. Every time we hit the refresh button, say in time $i_{\rm refresh}$, the adversary gets to see all of the outputs given by the algorithm {\em before} time $i_{\rm refresh}$. The adversary may use this in order to determine the next updates and queries. From that point on, until the next time we hit the refresh button, the adversary is oblivious in the sense that it does not get to see additional outputs of the algorithm.\footnote{In our application, we hit the refresh button after every $T$ outputs of the algorithm, which means that the adversary gets to see the outputs of the algorithm in bulks of size $T$.}
    \end{enumerate}

We show that if we have an algorithm $\A$ with a (hopefully fast) refresh button, then when applying our generic construction to $\A$, in Step~4 of the generic construction it suffices to {\em refresh} all the copies of algorithm $\A$, instead of completely resetting them. Assuming that refresh is indeed faster than reset, then we get a speedup to our resulting construction. We then show how to construct dynamic algorithms (for the graph problems we are interested in) with a fast refresh button. %

\paragraph{Designing algorithms with fast refresh.} Let $\A$ be an oblivious algorithm for one of our graph problems. In order to speedup $\A$'s reset time, we ideally would have used an appropriate sparsifier, and run our oblivious algorithm $\A$ on top of the sparsifier. This would make sure that the time needed to restart $\A$ (without restarting the sparsifier) depends on $n$ rather than $m$. Unfortunately, known sparsifiers that approximate well the functions that we estimate, do not work against an adaptive adversary. Consequently, if we use such a sparsifier (and never reset it) then the adversary may learn about the sparsifier's randomness through the estimates spitted out by the algorithm and use this knowledge to fool the algorithm. On the other hand, resetting the sparsifier would cost us $O(m)$ time, which would be too much.

To overcome this challenge, we design a sparsifier with a fast $\tilde{O}(n)$ time {\em refresh} button. Hitting the refresh button does not completely reset the sparsifier, but it does satisfy properties (i) and (ii), which suffice for our needs. We show that running an oblivious algorithm $\A$ on top of such a sparsifier results in an algorithm $\A'$ with a fast refresh operation: To refresh this algorithm $\A'$ we refresh the sparsifier and reset algorithm $\A$ (the reset time of  algorithm $\A$ would be fast since it is running on top of the sparsifier). Applying our generic construction to algorithm $\A'$ (and using its fast refresh operation in Step~4 of the generic construction) results in Theorem~\ref{thm:app}.

\subsubsection{Negative result for a search problem}
\label{sec:intro-search}

In \cref{sec:negative_search}, we prove that there is a search problem that 
is much easier for an oblivious algorithm than for an adaptive  algorithm.
We next describe the ideas of this proof. To simplify the presentation in the introduction, we fix some parameters; our separation in \cref{sec:negative_search} is stronger. 

Assume that $H_n:\set{0,1}^{n}\rightarrow \set{0,1}^n$ is a function such that $H_n(x)$  can be computed in time $O(n)$ for every $x \in \set{0,1}^n$.
Consider a dynamic problem in which an adversary maintains a set $X \subset \set{0,1}^n$ of \emph{excluded} strings (where $|X| < 2^{0.5n}$), initialized as the empty set. In each update the adversary adds an element to $X$ and the algorithm has to output the pair $\big(x,H_n(x)\big)$ for an element $x \notin X$. 
An oblivious dynamic algorithm picks a random $x$ in the preprocessing stage, computes $w=H_n(x)$, and outputs the pair $(x,w)$. No matter how an oblivious adversary chooses its updates to $X$, the probability that in a sequence of at most $2^{0.5n}$ updates the adversary adds to $X$ the random $x$ chosen by the algorithm  is negligible (as $x \in_R \set{0,1}^n$ and the size of $X$ is at most $2^{0.5n}$), and the algorithm can use the same output $(x,w)$ after each update, thus not paying at all for an update. However, an adaptive adversary can see 
the output $x_{i-1}$ of the algorithm after the $(i-1)$-th update and add $x_{i-1}$ to $X$. Thus, after each update the algorithm has to compute $H_n(x_i)$ for
a new $x_i$. We want to argue that an adaptive algorithm has to spend  $\Omega(n)$ amortized time per update. 
For this we need to assume that $H_n$ is moderately hard, e.g., computing $H_n(x)$ requires time $\Omega(n)$. However, this assumption does not suffice; 
we need to assume that computing
$H_n(x_1),\dots,H_n(x_\ell)$ for some $\ell$ requires time $\Omega(n\ell)$ for any sequence of inputs $x_1,\dots,x_\ell$ chosen by the algorithm.
That is, we need to assume that computing $H_n$ on many inputs cannot be done substantially  more efficiently than computing each $H_n(x_i)$ independently;
such assumption is known as a direct-sum assumption. Thus, assuming that there is a function that has a direct-sum property we get a separation. 

In the simple separation described above, the algorithm can in the preprocessing stage choose  a sequence $x_1,x_2,\dots,x_\ell$ 
and compute the values $H_n(x_1),H_n(x_2),\dots,H_n(x_\ell)$. Thereby  the algorithm does not need to spend any time after each update. To force the algorithm to work during the updates, in \cref{sec:negative_search} we define a more complicated dynamic problem and prove that the amortized update time of an adaptive algorithm for this problem is high even if the preprocessing time of the algorithm is $2^{0.5n}$.

The problem discussed above actually captures the very well-known technique in  dynamic graph algorithms for exploiting an oblivious adversary (e.g. how dynamic reachability and shortest paths algorithms choose a random root for an ES-tree \cite{roditty2008improved,henzinger2014decremental,bernstein2019decremental}, or how dynamic maximal matching algorithms choose a random edge to match \cite{BaswanaGS18,Solomon16}, and similarly for dynamic independent set \cite{ChechikZ19,BehnezhadDHSS19} and dynamic spanner algorithms \cite{BaswanaKS12}). 
Roughly speaking, the set $X$ above corresponds to a set of deleted edges in the graph. These algorithms need to ``commit'' to some $x$, but whenever $x$ is deleted/excluded from the graph, they need to recommit to a new $x'$ and spend a lot of time.
By choosing a random $x$, the algorithm would not recommit so often against an oblivious adversary, hence obtain small update time. %
Our result, therefore, formalizes the intuition that this general approach does not extend (at least as is) to the adaptive setting.

\subsubsection{Negative result for an estimation problem}

In Section~\ref{sec:negative_est} we present a separation result for an estimation problem. Our result uses techniques from the recent line of work on {\em adaptive data analysis} \cite{dwork2015preserving}. We remark that a similar connection to adaptive data analysis was utilized by \cite{KaplanMNS21}, in order to show impossibility results for adaptive streaming algorithms. However, our analysis differs significantly as our focus is on runtime lower bounds, while the focus of \cite{KaplanMNS21} was on space lower bounds.

To obtain our negative result, we introduce (and assume the existence of) a primitive, which we call {\em boxes scheme}, that allows a dealer to insert $m$ plaintext inputs into ``closed boxes'' such that:
\begin{enumerate}
    \item A closed box can be opened, retrieving the plaintext input, in time $T$ (for some parameter $T$).
    \item Any algorithm that runs in time $b\cdot T$ cannot learn ``anything'' on the content of more than $O(b)$ of the boxes.
\end{enumerate}

Given this primitive (which we define precisely in Section~\ref{sec:negative_est}), we consider the following problem.

\begin{definition}[The average of boxes problem, informal]
The initial input consists of $m$ closed boxes. On every time step, the algorithm gets a predicate (mapping {\em plaintexts} to $\{0,1\}$), and the algorithm needs to estimate the average of this predicate over the {\em content} of the $m$ boxes.
\end{definition}

This is an easy problem in the oblivious setting, because the algorithm can sample $\tilde{O}(1)$ of the boxes, open them, and use their content in order to estimate the average of all the predicates given throughout the execution. However, as we show, this is a hard problem in the adaptive setting. Specifically, every adaptive algorithm for this problem essentially must open $\Omega(m)$ of the $m$ boxes it gets as input. Intuitively, as opening boxes takes time, we get a separation between the oblivious and the adaptive settings, thereby proving item 2 of Theorem~\ref{thm:negative_intro}.

\section{Preliminaries on Differential Privacy}

Roughly speaking,
an algorithm is differentially private if its output distribution is ``stable'' w.r.t.\ a change to a single input element. To formalize
this, let $X$ be a domain. A \emph{database} $S\in X^{n}$ is a list
of elements from domain $X$. The \emph{$i$-th row }of $S$ is the
$i$-th element in $S$.
\begin{defn}
[Differential Privacy]A randomized algorithm $\A$ is $(\epsilon,\delta)$-\emph{differentially
private} (in short $(\epsilon,\delta)$\emph{-DP}) if for any two
databases $S$ and $S'$ that differ on one row and any subset of
outputs $T$, it holds that 
\[
\Pr[\A(S)\in T]\le e^{\epsilon}\cdot\Pr[\A(S')\in T]+\delta,
\]
where the probability is over the randomness of $\A$. The parameter $\epsilon$ is referred to as the \emph{privacy parameter}. When $\delta=0$ we omit it and write $\epsilon$-DP.
\end{defn}

\paragraph{Composition.} A crucial property of differential privacy is that it is preserved under
\emph{adaptive composition}. Let $\A_{1}$ and $\A_{2}$ be algorithms. The adaptive composition $\A=\A_{2}\circ\A_{1}$
is such that, given a database $S$, $\A$ invokes $a_1=\A_1(S)$, then $a_2=\A_2(a_1,S)$, and finally outputs $(a_1,a_2)$. The \emph{basic
composition theorem} guarantees that, if $\A_{1},\dots,\A_{k}$ are
each $(\epsilon,\delta)$-DP algorithms, then the composition $\A_{k}\circ\dots\circ\A_{1}$
is $(\epsilon'=k\epsilon,k\delta)$-DP. The \emph{advanced composition theorem} shows that the privacy parameter $\epsilon'$ need not grow linearly in $k$, but instead only
in $\approx\sqrt{k}$. 

\begin{theorem}
[Advanced Composition \cite{DworkRV10}]\label{thm:composition}Let
$\epsilon,\delta'\in(0,1]$ and $\delta\in[0,1]$. If $\A_{1},\dots,\A_{k}$
are each $(\epsilon,\delta)$-DP algorithms, then $\A_{k}\circ\dots\circ\A_{1}$
is $(\epsilon',\delta'+k\delta)$-DP where 
\[
\epsilon'=\sqrt{2k\ln(1/\delta')}\cdot\epsilon+2k\epsilon^{2}.
\]
\end{theorem}

In our applications, the second term (which is linear in $k$) will be dominated by the first term. The
saving in $\epsilon'$ from $k$ to $\sqrt{k}$ enables a polynomial
speed up in our applications. 

\paragraph{Amplification via sampling.} \emph{Secrecy-of-the-sample} is a technique for ``amplifying'' privacy by subsampling. Informally, if a $\gamma$-fraction
of the input database rows are sampled, and only those are given as input to a differentially private algorithm then the privacy parameter is reduced (i.e., improved) by a factor proportional to $\approx\gamma$.
\begin{theorem}
[Amplicafication via sampling (Lemma 4.12 of \cite{Bun2015differentially})]\label{thm:subsample}Let
$\A$ be an $\epsilon$-DP algorithm where $\epsilon\le 1$. Let
$\A'$ be the algorithm that, given a database $S$ of size $n$,
first constructs a database $T\subset S$ by sub-sampling with repetition
$k\le n/2$ rows from $S$ and then returns $\A(T)$. Then, $\A'$ is $(\frac{6k}{n}\cdot\epsilon)$-DP.\footnote{The statement in \cite{Bun2015differentially} is more general
and allows $\A$ to be $(\epsilon,\delta)$-DP.} 
\end{theorem}

\paragraph{Generalization.} Our analysis relies on the \emph{generalization
property} of differential privacy. 
Let
${\cal D}$ be a distribution over a domain $X$ and let $h:X\rightarrow\{0,1\}$
be a predicate. 
Suppose that the goal is to estimate $h({\cal D})=\Ex_{x\sim{\cal D}}[h(x)]$. A simple
solution is to sample a set $S$ consisting of few elements from $X$ independently from
${\cal D}$,  and then compute the empirical average
$h(S)=\frac{1}{|S|}\sum_{x\in S}h(x)$. By
standard concentration bounds, we have $h({\cal D})\approx h(S)$.
That is, the empirical estimate on a small sample $S$ \emph{generalizes}
to the estimate over the underlying distribution ${\cal D}$. 

The argument above, however, fails if $S$ is sampled first
and $h$ is chosen adaptively, because $h$ can ``overfit'' $S$.
The theorem below says that, as long as the predicate $h$ is generated
from a differentially private algorithm $\A$, we can still guarantee generalization of $h$ (even when the choice of $h$ is a function of $S$). As
shown in \cite{HassidimKMMS20}, this key property will link
differentially privacy to accuracy of algorithms against an adaptive
adversary. 
\begin{theorem}
[Generalization of DP \cite{dwork2015preserving,bassily2021algorithmic}]\label{thm:generalize}Let
$\epsilon\in(0,1/3)$, $\delta\in(0,\epsilon/4)$ and $t\ge\frac{1}{\epsilon^{2}}\log(\frac{2\epsilon}{\delta})$.
Let ${\cal D}$ be a distribution on a domain $X$. Let $S\sim{\cal D}^{t}$
be a database containing $t$ elements sampled independent from ${\cal D}$.
Let ${\cal A}$ be an algorithm that, given any database $S$ of
size $t$, outputs a predicate $h:X\rightarrow\{0,1\}$. 
(We emphasize that $h$ may depend on $S$.)

If $\A$ is $(\epsilon,\delta)$-DP, then the empirical average of
$h$ on sample $S$, i.e., $h(S)=\frac{1}{|S|}\sum_{x\in S}h(x)$, and
$h$'s expectation over the underlying distribution ${\cal D}$, i.e.,
$h({\cal D})=\Ex_{x\sim{\cal D}}[h(x)]$, are within $10\epsilon$ with probability
at least $1-\frac{\delta}{\epsilon}$. In other words, we have 
\[
\Pr_{\underset{h\gets\A(S)}{S\sim{\cal D}^{t}}}\left[\left|\frac{1}{|S|}\sum_{x\in S}h(x)-\Ex_{x\sim{\cal D}}[h(x)]\right|\ge10\epsilon\right]<\frac{\delta}{\epsilon}.
\]
\end{theorem}

The only differentially private subroutine we need is a very simple
algorithm for computing an approximate median of elements in databases. 

\begin{theorem}
[Private Median  (folklore)]\label{thm:primed}Let
$X$ be a finite domain with total order. For every $\epsilon,\beta\in(0,1)$,
there is $\Gamma=O(\frac{1}{\epsilon}\log(\frac{|X|}{\beta}))$ such
that the following holds. There exists an $(\epsilon,0)$-DP algorithm
$\primed_{\epsilon,\beta}$ that, given a database $S\in X^{*}$, 
in $O(|S|\cdot\frac{1}{\epsilon}\log^{3}(\frac{|X|}{\beta})\cdot\polylog(|S|))$
time outputs an element $x\in X$ (possibly $x\notin S$) such that,
with probability at least $1-\beta$, there are at least $|S|/2-\Gamma$
elements in $S$ that are bigger or equal to $x$ and at least $|S|/2-\Gamma$
elements in $S$ that are smaller or equal to $x$.
\end{theorem}

The algorithm is based on binary search. If we assume that we can sample a real number from  the Laplace distribution in constant time, then running time would be $O(|S|\log|X|)$. Here, we do not assume that and use the bound from \cite{BalcerV19}. We remark that there are several advanced constructions for private median with error that grows very slowly as a function of the domain size $|X|$ (only polynomially with $\log^{*}|X|$). \cite{BeimelNS16,Bun2015differentially,KaplanLMNS20} In our application, however, the domain size is already small, and hence, we can use simpler construction stated above.

\section{A Generic Reduction}

In this section, we present a simple black-box transformation of dynamic
algorithms against an oblivious adversary to ones against an adaptive
adversary. The approach builds on the work of \cite{HassidimKMMS20}
for transforming streaming algorithms, which focus on space 
instead of update time. A key difference is that we apply subsampling
for speeding up.
This is crucial for enabling our applications in \Cref{sec:dynamic app} have non-trivial update and query time. %
We give the formal statement and its proof below.
\begin{theorem}
\label{thm:basic} \label{cor:basic}
Let $\delta_\fail,\alpha>0$ be parameters.
Let $g$ 
be a function that maps elements in some domain $X$ to a number in
$[-U,-\frac{1}{U}]\cup\{0\}\cup[\frac{1}{U},U]$ where $U>1$. Suppose
there is a dynamic algorithm $\A$ against an \emph{oblivious} adversary
that, given an initial input $x_{0}\in X$ undergoing a sequence of $T$
updates, guarantees the following:
\begin{itemize}
\item The total preprocessing and update time for handling $T$ updates
is $t_{\total}$. 
\item The query time is $t_{q}$ and, with probability $\ge9/10$, the answer
is a $\gamma$-approximation of $g(x)$.
\end{itemize}
Then, there is a dynamic algorithm $\A'$ against an \emph{adaptive}
adversary that, with probability at least $1-\delta_{\fail}$, maintains
a $\gamma(1+\alpha)$-approximation of a function $g(x)$ when the
input undergoes $T$ updates in $\Otil(\sqrt{T}\cdot t_{\total}+T\cdot t_{q})$
total update time.
By restarting the algorithm every $T$ steps we can run $\A'$ 
on a sequence of updates of any length in $
\Otil(t_{\total}/\sqrt{T}+t_{q})
$ amortized update time.
The $\Otil$, in this section, hides $\polylog(\frac{T\log  U}{\alpha\delta_{\fail}})$ factors.
\end{theorem}

\paragraph{Algorithm $\A'$ Description.}

Before the first update arrives, we initialize $c=\Otil(\sqrt{T})$
copies of $\A$, denoted by $\A^{(1)},\dots,\A^{(c)}$, with the input
$x_{0}$. For each step $i\in[\ibegin,\ifin]$, given an update $u_{i}$,
we feed the update $u_{i}$ to every copy of $\A$. Then, to
compute a $\gamma(1+\alpha)$-approximation of $g(x_{i})$, we independently uniformly sample $s=\Otil(1)$ indices $j_{1},\dots,j_{s}\in[1,c]$.\footnote{This is unlike in \cite{HassidimKMMS20} where all instances of $\A$ are used.} We emphasize the these sampled indices are not revealed to the adversary.
For each $1\le k\le s$, we query $\A^{(j_{k})}$ and let $\output_{i}^{(j_{k})}$
denote its estimate of $g(x_{i})$. Then, we round up $\output_{i}^{(j_{k})}$
to the nearest power of $(1+\alpha$) and denote it by $\outtil_{i}^{(j_{k})}$.
If $\output_{i}^{(j_{k})}=0$, we set $\outtil_{i}^{(j_{k})}=0$.
Finally, $\A'$ outputs the estimate of $g(x_{i})$ as
\[
\output'_{i}=\primed_{\epsilon_{\med},\beta}(\outtil_{i}^{(j_{1})},\dots,\outtil_{i}^{(j_{s})})
\]
where $\primed_{\epsilon_{\med},\beta}$ is for estimating the median with differential privacy
(see \Cref{thm:primed}) and $\epsilon_{\med}=1/2$ and $\beta=\delta_{\fail}/2T$. 

\paragraph{Specifying Parameters.}

Here, we specify the parameters of the algorithm. This is needed for precisely stating the total update time. 
Let $X_{\med}$ denote the total ordered domain for $\primed_{\epsilon_{\med},\beta}$.
In our setting, $X_{\med}$ is simply the range of $\outtil_{i}^{(j_{k})}$
which is $\{0\}\cup\{\pm(1+\alpha)^{a}\mid-\log_{(1+\alpha)}U\le a\le\log_{(1+\alpha)}U\}$.
So $|X_{\med}| = O(\frac{\log U}{\alpha})$. The parameter $\Gamma$
from \Cref{thm:primed} is such that $\Gamma=O(\frac{1}{\epsilon_{\med}}\log(\frac{|X_{\med}|}{\beta}))=O(\log(\frac{\log U}{\alpha\beta}))=\Otil(1)$. 
Now, we set $s = 100\Gamma$, so that when $\primed_{\epsilon_\med,\beta}$ is given $s$ numbers, it returns a number whose rank is in $s/2 \pm \Gamma = (1/2  \pm 1/100)s$, a good enough approximation of the median. Lastly, we set 
\begin{align*}
c & =\paramc=\Otil(\sqrt{T}).
\end{align*}
This choice of $c$ is such that after subsampling and composition the entire algorithm  would be private for the appropriate parameters with respect to its random bits. See  Corollary \ref{cor:transcript whole}.

\paragraph{Total Update Time.}

The total time $c$ copies of $\A$ preprocess the initial input $x_{0}$
and handle all $T$ updates is clearly $c\cdot t_{\total}$ by definition
of $t_{\total}$. For each step $i$, we only query $s$ many copies
of $\A$ to obtain $\output_{i}^{(j_{1})},\dots,\output_{i}^{(j_{s})}$.
Rounding $\output_{i}^{(j_{k})}$ to its nearest power of $(1+\alpha)$
takes $O(\log\frac{\log U}{\alpha})$ time by binary search. Computing
$\output'_{i}$ is to evaluate $\primed_{\epsilon_{\med},\beta}$
which takes $t_{\med}=O(s\cdot\frac{1}{\epsilon_{\med}}\log(\frac{|X_{\med}|}{\beta})\cdot\polylog(s))=O(s\cdot\frac{1}{\epsilon_{\med}}\log(\frac{\log U}{\alpha\beta})\cdot\polylog(s))$
time by \Cref{thm:primed}. Therefore, we can conclude:
\begin{prop}
\label{prop:update time}The total update time of $\A'$ is at most
\begin{align*}
 & c\cdot t_{\total}+T\times O\left((t_{q}+\log\frac{\log U}{\alpha})\cdot s+t_{\med}\right)\\
 & =O\left(t_{\total}\cdot\sqrt{T}\cdot\log\left(\frac{T\log U}{\alpha\delta_{\fail}}\right)\cdot\sqrt{\log\frac{T}{\delta_{\fail}}}\right)+O\left(t_{q}\cdot T\cdot\log\left(\frac{T\log U}{\alpha\delta_{\fail}}\right)+T\cdot \polylog\left(\frac{T\log U}{\alpha\delta_{\fail}}\right)\right)\\
 & =\Otil(t_{\total}\sqrt{T}+t_{q}T)
\end{align*}
\end{prop}
The second line is obtained by simply plugging in the definitions of $\beta = \delta_\fail/2T$, 
$\Gamma = O(\log(\frac{\log U}{\alpha\beta})) = O(\log(\frac{T\log U}{\alpha\delta_{\fail}}))$,
$s = 100\Gamma = O(\log(\frac{T\log U}{\alpha\delta_{\fail}}))$ and, hence, we have that $c = O(\sqrt{T}\cdot\log(\frac{T\log U}{\alpha\delta_{\fail}})\cdot\sqrt{\log\frac{T}{\delta_{\fail}}})$ and so the second line follows.

\paragraph{Accuracy Against an Adaptive Adversary.}

It only remains to argue that $\A'$ maintains accurate approximation
of $g(x)$ against  an adaptive adversary. Let $r^{(1)},\dots,r^{(c)}\in\{0,1\}^{*}$
denote the random strings used by the oblivious algorithms $\A^{(1)},\dots,\A^{(c)}$
during the $T$ updates. We view the collection of random string $R=\{r^{(1)},\dots,r^{(c)}\}$
as a database where each $r^{(j)}$ is its row. We will show that
the transcript of the interaction between the adversary and algorithm
$\A'$ is differentially private \emph{with respect to $R$.} (This
is perhaps the most important conceptual idea from \cite{HassidimKMMS20}.)
Then, we will exploit this fact to argue that the answers of $\A'$
are accurate. Let us formalize this plan below.

For any time step $i$, let $\output'_{i}(R)$ denote the output of
$\A'$ at time step $i$ when the collection $R$ is fixed. Note that
$\output'_{i}(R)$ is still a random variable because $\A'$ uses
some additional random strings for subsampling and computing a private
median. Now, we define ${\cal T}_{i}(R)=(u_{i},\output'_{i}(R))$
as the transcript between $\adversary$ and algorithm $\A'$ at step
$i$. Let 
\[
{\cal T}(R)=x_{0},{\cal T}_{\ibegin}(R),\dots,{\cal T}_{\ifin}(R)
\]
denote the transcript. We also prepend the transcript with the input
$x_{0}$ before the first update arrives. Since $R$ is freshly sampled
at the beginning, it is completely independent from $x_{0}$.
We view ${\cal T}_{i}$ and ${\cal T}$ as algorithms that, given
a database $R$, return the transcripts. From this view, we can prove
that they are differentially private with respect to $R$.
\begin{lem}
\label{lem:transcript round}For a fixed step $i$, ${\cal T}_{i}$
is $(\frac{6s}{c}\cdot\epsilon_{\med},0)$-DP with respect to $R$. 
\end{lem}

\begin{proof}
Given a transcript ${\cal T}_{i}(R)=(u_{i},\output'_{i}(R))$ of only
a single step $i$, $u_{i}$ does not give any (new) information about $R$.
So it suffices to consider $\output'_{i}(R)$ which is set to $\primed_{\epsilon_{\med},\beta}(\outtil_{i}^{(j_{1})},\dots,\outtil_{i}^{(j_{s})})$.
By \Cref{thm:primed}, we have that $\primed_{\epsilon_{\med},\beta}$
is $(\epsilon_{\med},0)$-DP. Its inputs are $\outtil_{i}^{(j_{1})},\dots,\outtil_{i}^{(j_{s})}$
which are determined by the subset $\{r^{(j_{1})},\dots,r^{(j_{s})}\}\subset R$,
which in turn are obtained by sub-sampling from $R$. By invoking
\Cref{thm:subsample} (and note that $s\le c/2$), subsampling boosts
the privacy parameter and so $\output'_{i}(R)$ is $(\frac{6s}{c}\cdot\epsilon_{\med},0)$-DP
with respect to $R$ as claimed. 
\end{proof}
\begin{cor}
\label{cor:transcript whole}${\cal T}$ is $(\frac{1}{100},\frac{\beta}{100})$-DP
with respect to $R$.
\end{cor}

\begin{proof}
Observe that ${\cal T}$ is an adaptive composition ${\cal T}_{\ifin}\circ\dots\circ{\cal T}_{\ibegin}$
(except that we prepend $x_{0}$ which is independent from $R$).
Since each ${\cal T}_{i}$ is $(\frac{6s}{c}\cdot\epsilon_{\med},0)$-DP
as shown in \Cref{lem:transcript round}, by applying the advanced
composition theorem (\Cref{thm:composition}) with parameters $\epsilon=\frac{6s}{c}\epsilon_{\med}$,
$\delta=0$, and $\delta'=\beta/100$, we have that that ${\cal T}$
is $(\epsilon',\delta k+\delta')$-DP where 
\begin{align*}
\epsilon' & =\sqrt{2T\ln(100/\beta)}\cdot\left(\frac{6s}{c}\epsilon_{\med}\right)+2T\cdot\left(\frac{6s}{c}\epsilon_{\med}\right)^{2}\\
 & \le\frac{1}{200}+\frac{1}{200}=\frac{1}{100}
\end{align*}
because $c=\paramc$. Also, $\delta k+\delta'=\beta/100$. Therefore,
${\cal T}$ is $(\frac{1}{100},\frac{\beta}{100})$-DP.
\end{proof}
Next, we exploit differential privacy for accuracy against an
adaptive adversary. Let $\xvec_{i}=(x_{0},u_{\ibegin},\dots,u_{i})$
denote the whole input sequence up to time $i$. Let $\A(r,\xvec_{i})$
denote the output of the oblivious algorithm $\A$ on input sequence
$\xvec_{i}$, given a random string $r$. Let 
\[
\acc_{\xvec_{i}}(r)=\one\left\{ g(x_{i})\le\A(r,\xvec_{i})\le\gamma g(x_{i})\right\} 
\]
be the indicator function deciding if $\A(r,\xvec_{i})$ is $\gamma$-accurate.
Note that $\acc_{\xvec_{i}}(r^{(j)})$ indicates precisely whether
the instance $\A^{(j)}$ is accurate at time $i$. Now, we show that
at all times, most instances of the oblivious algorithm are $\gamma$-accurate. 
\begin{lem}
\label{lem:mostly accurate}For each fixed $i\in[\ibegin,\ifin]$,
$\sum_{j=1}^{c}\acc_{\xvec_{i}}(r^{(j)})\ge\frac{4}{5}c$ with probability
at least $1-\beta$.
\end{lem}

\begin{proof}
Observe that the function $\acc_{\xvec_{i}}(\cdot)$ is determined
by the transcript ${\cal T}$. This is because the input sequence
$\xvec_{i}$ is just a substring of the transcript ${\cal T}$ and
$\xvec_{i}$ determines the predicate $\acc_{\xvec_{i}}$. 

Now, we have the following (1) each row of $R$ is a string drawn
independently from the uniform distribution ${\cal U}$, (2) $\acc_{\xvec_{i}}$
is a predicate on strings and is determined by ${\cal T}$ as argued
above, and (3) ${\cal T}$ can be viewed as a $(\frac{1}{100},\frac{\beta}{100})$-DP
algorithm with respect to $R$ by \Cref{cor:transcript whole}. By
the generalization property of differential privacy (\Cref{thm:generalize}),
we have that the empirical average of $\acc_{\xvec_{i}}$ on $R$,
i.e., $\frac{1}{c}\sum_{j=1}^{c}\acc_{\xvec_{i}}(r^{(j)})$, and $\acc_{\xvec_{i}}$'s
expectation over the underlying distribution ${\cal U}$, i.e., $\Ex_{r\sim{\cal U}}[\acc_{\xvec}(r)]$
should be close to each other. More formally, by invoking \Cref{thm:generalize}
where $\epsilon=1/100$, $\delta=\beta/100$ and $t=c\gg\frac{1}{\epsilon^{2}}\log(\frac{2\epsilon}{\delta})$,
we have that 

\[
\Pr_{\underset{\acc_{\xvec_{i}}\gets{\cal T}(R)}{R\sim{\cal U}^{c}}}\left[\left|\frac{1}{c}\sum_{j=1}^{c}\acc_{\xvec_{i}}(r^{(j)})-\Ex_{r\sim{\cal U}}[\acc_{\xvec_{i}}(r)]\right|\ge\frac{1}{10}\right]\le\beta.
\]

Since the oblivious algorithm $\A$ returns accurate answers with
probability at least $9/10$ as long as its random choice is independent
from the input, for any arbitrary input sequence $\xvec$, we have
$\Ex_{r\sim{\cal U}}[\acc_{\xvec}(r)]\ge9/10$. Therefore, we have
that with probability at least $1-\beta$, 
\[
\frac{1}{c}\sum_{j=1}^{c}\acc_{\xvec_{i}}(r^{(j)})\ge\frac{9}{10}-\frac{1}{10}\ge\frac{4}{5}
\]
as desired. 
\end{proof}
Given that most instances of the oblivious algorithm are always accurate,
it is intuitively immediate that $\A'$ is always accurate too. This
is because $\A'$ returns the median of the sub-sampled answers from
oblivious algorithms. Below, we verify this.

\begin{cor}
\label{cor:accurate whole}For all $i\in[\ibegin,\ifin]$, $g(x_{i})\le\output'_{i}\le\gamma(1+\alpha)g(x_{i})$
with probability at least $1-\delta_{\fail}$.
\end{cor}

\begin{proof}
Consider a fixed step $i$. Recall that $\A'$ independently samples
$s$ indices $j_{1},\dots,j_{s}$ and queries $\A^{(j_{k})}$ for
$1\le k\le s$. Let $\acc_{k}=\acc_{\xvec_{i}}(r^{(j_{k})})$ indicate
whether $\A^{(j_{k})}$ is accurate at time $i$. \Cref{lem:mostly accurate}
implies that $\Ex[\sum_{k=1}^{s}\acc_{k}]\ge\frac{4}{5}s$. By Hoeffding's
bound, $\sum_{k=1}^{s}\acc_{k}\ge\frac{3}{4}s$ with probability at
least $1-\exp(-\Theta(s))\ge1-\beta$ by making sure that the constant
in the definition of $s$ (actually $\Gamma$) is large enough. 

If $\acc_{k}=1$, we have that $g(x_{i})\le\output_{i}^{(j_{k})}\le\gamma g(x_{i})$
and so $g(x_{i})\le\outtil_{i}^{(j_{k})}\le\gamma(1+\alpha)g(x_{i})$.
So at least $\frac{3}{4}$-fraction of $\outtil_{i}^{(j_{1})},\dots,\outtil_{i}^{(j_{s})}$
are $\gamma(1+\alpha)$-approximation of $g(x_{i})$. With probability
at least $1-\beta$, $\primed_{\epsilon_{\med},\beta}$ returns $\output'_{i}$
such that there are $\frac{1}{2}-\frac{\Gamma}{s}\ge\frac{49}{100}$
fraction of $\outtil_{i}^{(j_{1})},\dots,\outtil_{i}^{(j_{s})}$ that
are at least $\output'_{i}$ and the same holds for those that are
at most $\output'_{i}$. Therefore, $\output'_{i}$ is a $\gamma(1+\alpha)$-approximation
of $g(x_{i})$ with probability at least $1-2\beta$. By union bound,
this holds over all time steps with probability at least $1-2T\beta=1-\delta_{\fail}$.
\end{proof}
Via the accuracy guarantee from \Cref{cor:accurate whole} together
with the total update time bound from \Cref{prop:update time}, we
now conclude the proof of \Cref{thm:basic}.

\subsection*{Extensions of \Cref{thm:basic}}

Before we conclude this section, we discuss several possible ways to
extend the reduction from \Cref{thm:basic}. 

\paragraph{Worst-case update time.}

First of all, although the reduction is stated for amortized update
time, it can be made worst-case if the given oblivious algorithm guarantee
worst-case update time. More formally, we have the following.
\begin{theorem}
\label{thm:basic worst-case}Let $g:X\rightarrow[-U,-\frac{1}{U}]\cup\{0\}\cup[\frac{1}{U},U]$
be a function that maps elements in some domain $X$ to a number in
$[-U,-\frac{1}{U}]\cup\{0\}\cup[\frac{1}{U},U]$ where $U>1$. Suppose
there is a dynamic algorithm $\A$ against an \emph{oblivious} adversary
that, given an initial input $x_{0}$ undergoing a sequence of $T$
updates, guarantees the following:
\begin{itemize}
\item The preprocessing time on $x_{0}$ is $t_{p}$
\item The worst-case update time for each update is $t_{u}$. 
\item The query time is $t_{q}$ and, with probability $\ge9/10$, the answer
is a $\gamma$-approximation of $g(x)$.
\end{itemize}
Then, there is a dynamic algorithm $\A''$ against an \emph{adaptive}
adversary that, with probability at least $1-\delta_{\fail}$, maintains
a $\gamma(1+\alpha)$-approximation of a function $g(x)$ when $x$
undergoes any sequence of update using 
\[
\Otil\left(\frac{t_{p}}{\sqrt{T}}+\sqrt{T}t_{u}+t_{q} \right)
\]
 worst-case update time. 
\end{theorem}

\begin{proof}
[Proof sketch]Using exactly the same algorithm from \Cref{thm:basic},
we obtain the adaptive algorithm $\A'$ can handle $T$ updates whose
preprocessing time is $\Otil(t_{p}\sqrt{T})$ and the worst-case update/query
time is $\Otil(t_{u}\sqrt{T}+t_{q})$. To get an algorithm with $\Otil(\frac{t_{p}}{\sqrt{T}}+\sqrt{T}t_{u}+t_{q})$
worst-case update time, we create two instances $\A'_{odd}$ and $\A'_{even}$
of $\A'$ and proceed in phases. Each phase has $\Theta(T)$ updates.
We only need to show how to avoid spending a large preprocessing time
of $\Otil(t_{p}\sqrt{T})$ in a single time step. During the odd phases,
we use $\A'_{odd}$ to handle the queries and distribute the work
for preprocessing of $\A'_{even}$ equally on each time step in this
phase. During the even phases, we do the opposite. So the preprocessing
time is ``spread'' over $\Theta(T)$ updates, which consequently
contributes $\Otil(\frac{t_{p}}{\sqrt{T}})$ worst-case update time.
This a very standard technique in the dynamic algorithm literature
(see e.g.~Lemma 8.1 of \cite{baswana2016dynamic}). 
\end{proof}

\paragraph{Speed up for stable answers.}

Suppose that we know that during $T$ updates, the $(1+\epsilon)$-approximate
answers to the queries can change only $\lambda$ times for some $\lambda\ll T$.
Then, using the same idea from \cite{HassidimKMMS20}, the total update
time of \Cref{thm:basic} can be improved to $\Otil(t_{\total}\sqrt{\lambda}+t_{q})$. 

This idea could be useful for several problems. For examples, suppose
that we want to maintain $(1+\epsilon)$-approximation of global minimum
cut, or $(s,t)$-minimum cuts, or maximum matching in unweighted graphs.
If the current answer $k$, we know that it must takes at least $\epsilon k$
updates before the answer changes by a $(1+\epsilon)$ factor. Suppose
that, somehow, the answer is always at least $k$, then we have $\lambda\le T/\epsilon k$.
It is also possible to remove assumption that the answer is at least
$k$: if there is a separated adaptive algorithm that can take care
of the problem when the answer is less than $k$, then we can run
both algorithms in parallel. This idea of combining the two algorithms,
one for small answers and another for large answers, was explicitly
used in \cite{ben2021adversarially} in the streaming setting.

\paragraph{Speed up via batch updates.}

In the algorithm for \Cref{thm:basic}, recall that we create $c=\Otil(\sqrt{T})$
copies of $\A$, denoted by $\A^{(1)},\dots,\A^{(c)}$. For each copy
$\A^{(j)}$, we feed an update $u_{i}$ at every time step $i$ one
by one. Consider what if we are lazy in feeding the update to $\A^{(j)}$.
That is, we wait until $\A^{(j)}$ is sampled and we want to query
$\A^{(j)}$. Only then we feed a batch of updates containing all updates
that we have not feed to $\A^{(j)}$ until the current update. The
batch will be of size $O(c/s)=\Otil(\sqrt{T})$ in expectation. That
is, through out the sequence of $T$ updates, each $\A^{(j)}$ is
expected to handles only $\Otil(\sqrt{T})$ batches containing $\Otil(\sqrt{T})$
updates.

This implementation would not change the correctness. But the whole
algorithm can possibly be faster if $\A$ is a dynamic algorithm that
can handle a batch update of size $\Otil(\sqrt{T})$ faster than a
sequence of $\Otil(\sqrt{T})$ updates. This is a property that is
quite natural to expect. Indeed, there are some dynamic algorithms
such that the larger the batch the faster the update time on average,
including dynamic matrix inverse and its applications such that dynamic
reachability \cite{sankowski2010fast,van2019dynamic}. However, these
algorithms are not for approximating functions and so we cannot exploit
them in this paper.

\section{Applications to Dynamic Graph Algorithms }

\label{sec:dynamic app}

In this section, we show new dynamic approximation algorithms against
an adaptive adversary for four graph problems including, global minimum
cut, all-pairs distances, effective resistances, and minimum cuts.
Given a graph $G$ undergoing edge updates, let $n$ denote a number
of vertices and $m$ denote a \emph{current} number of edges in $G$.
In \Cref{sec:apply reduction}, we show how the generic reduction from
\Cref{thm:basic} immediately transforms known algorithms against an
oblivious adversary to work against an adaptive adversary with $o(m)$
update and query time. Then, in \Cref{sec:sparsification period},
we show a sparsification technique which allows us to assume that
$m=\Otil(n)$ all the time and speed up all of our algorithms.

Throughout this section, $\Otil$ hides a $\polylog(n)$ factor. We
also assume edge weights are integers of size at most $\poly(n)$.
Also, to simplify the calculation, we will assume $\epsilon=\Omega(1)$
in all of our dynamic $(1+\epsilon)$-approximation algorithms.%

\subsection{Applying the Generic Reduction}

\label{sec:apply reduction}

The update time of our dynamic algorithms in this subsection depends
on $m$. Although we will later show in \Cref{sec:sparsification period}
that we can assume $m=\Otil(n)$, even within this subsection, we
can also assume that $m$ never changes by more than a constant factor,
because otherwise, we can restart the algorithm from scratch which
would increase the amortized update time by at most a constant factor. 

We start with the first dynamic algorithm against an adaptive adversary
for $(1+\epsilon)$-approximate global mincut. 
\begin{cor}
[Global minimum cuts]
\label{cor:globalmincut}
There is a dynamic
algorithm against an adaptive adversary that, given a weighted graph
$G$ undergoing edge insertions and deletions, with probability $1-1/\poly(n)$,
maintains a $(1+\epsilon)$-approximate value of the global mincut
in $\Otil(m^{1/2}n^{1/4})=\Otil(m^{3/4})$ amortized update time.
\end{cor}

\begin{proof}
We simply apply \Cref{thm:basic} to the dynamic algorithm against an oblivious adversary
by Thorup \cite{thorup2007fully} (Theorem 11).
When the graph initially has $m$ weighted edges, his algorithm takes
$\Otil(m)$ preprocessing time\footnote{The preprocessing time is not explicitly stated in \cite{thorup2007fully}.
This preprocessing includes, graph sparsification via uniform sampling,
greedily packing $\Otil(1)$ trees, and initializing information inside
each tree. All of these takes near-linear time.} and $\Otil(\sqrt{n})$ worst-case update time. So the total update
time for handling
$T$, 
updates (for any $T$) is 
\[
t_{\total}=\Otil(m+T\sqrt{n}).
\]

Thorup's algorithm maintains the $(1+\epsilon/3)$-approximation of
of the global mincut explicitly, so we can query it in $t_{q}=O(1)$
time. By plugging this into \Cref{cor:basic} where $\alpha=\epsilon/3$,
since $(1+\epsilon/3)\cdot(1+\alpha)\le(1+\epsilon)$, we obtain an
$(1+\epsilon)$-approximation algorithm against an adaptive adversary
with amortized update time
\[
\Otil\left(\frac{m+T\sqrt{n}}{\sqrt{T}}\right).
\]
This amortized update time is minimized for 
$T\approx m/\sqrt{n}$.
Therefore if we rebuild our data structure after
$T\approx m/\sqrt{n}$ updates we get an amortized update time of 
$\Otil(m^{1/2}n^{1/4})$.
\end{proof}

\begin{cor}
[All-pairs distances]\label{cor:distance accurate slow}There is
a dynamic algorithm against an adaptive adversary that, given a weighted
graph $G$, handles the following operations in $\Otil(m^{4/5})$
amortized update time:
\begin{itemize}
\item Insert or delete an edge from the graph, and
\item Given $s,t\in V(G)$, with probability at least $1-1/\poly(n)$, return
a $(\log(n)\cdot\poly(\log\log n))$-approximation of the distance between $s$ and $t$. 
\end{itemize}
\end{cor}

\begin{proof}
Chen et al.~\cite{chen2020fast} (Lemma 7.15 and the proof of Theorem
7.1) gives a dynamic algorithm $\A$ against an oblivious adversary
with the following guarantee. Given a weighted graph $G$ with $n$
vertices and $m$ edges and any parameter $j$, the algorithm preprocesses
$G$ and with probability $1-1/\poly(n)$ handles at most $T=O(j)$
operations in $t_{\total}=\Otil(m^{2}/j)$ total update time. 
The operations that $\A$ can handle include:
\begin{itemize}
\item edge insertions and deletions, and 
\item given $(s,t)$, return a $(\log(n)\cdot\poly(\log\log n))$-approximation
of the $(s,t)$-distance in $t_{q}=\Otil(j)$ time.  
\end{itemize}
We want to apply the transformation from \Cref{thm:basic} to $\A$,
but there is a small technical issue. In \Cref{thm:basic}, we only
consider algorithms that maintain one single number, but $\A$ can
return an answer for any pair $(s,t)$. So we instead assume that
$\A$ also maintains a pair of variables $(\texttt{src},\texttt{snk})$.
Each $(s,t)$-query to $\A$ contains two sub-steps: first we update
$(\texttt{src},\texttt{snk})\gets(s,t)$ and then $\A$ returns the
answer for $(\texttt{src},\texttt{snk})$, which is now the only single
number that $\A$ maintains. So we can indeed apply \Cref{thm:basic}
to $\A$. 

Applying \Cref{cor:basic} to $\A$,
we obtain an
algorithm against an adaptive adversary with amortized update time
of 
\[
\Otil\left(\frac{m^{2}/j}{\sqrt{j}}+j\right).
\]
By choosing $j=m^{4/5}$ (and restarting after every $j$ updates) we get an amortized update time of
$\Otil(m^{4/5})$.
\end{proof}
Next, we show that by plugging another oblivious algorithm into the
reduction, we can speed up the above result. 
\begin{cor}
[All-pairs distances with cruder approximation]\label{cor:distance crude slow}For
any integer $i\ge2$, there is a dynamic algorithm against an adaptive
adversary that, given a weighted graph $G$, handles the following
operations in $m^{1/2+1/(2i)+o(1)}$ amortized update time:
\begin{itemize}
\item Insert or delete an edge from the graph, and
\item Given $s,t\in V(G)$, with probability at least $1-1/\poly(n)$, return
a $O(\log^{3i-2} n)$-approximation of the distance between $s$ and
$t$.
\end{itemize}
\end{cor}

\begin{proof}
Forster, Goranci, and Henzinger \cite{forsterGH21} (Theorem 5.1) show
a dynamic algorithm $\A$ against an oblivious adversary that can
handle edge insertions and deletions in $m^{1/i+o(1)}$ amortized
update time \emph{when an initial graph is an empty graph} and, given
$s,t\in V(G)$, can return a $O(\log^{3i-2} n)$-approximate $(s,t)$-distance
in $\polylog(n)$ time. Again, we can use the same small modification
as in \Cref{cor:distance accurate slow} to view as $\A$ as an algorithm that maintain
only one number.

When an initial graph is not empty but has $m$ edges, algorithm $\A$
can handle $T$ operations of edge updates and queries in at most
$t_{\total}=(m+T)\cdot m^{1/i+o(1)}$ time with query time is $t_{q}=\polylog(n)$.
Therefore, via \Cref{thm:basic} we obtain an algorithm $\A'$ against
an adaptive adversary with amortized update time of 
\[
\Otil\left(\frac{(m+T)\cdot m^{1/i+o(1)}}{\sqrt{T}}+\polylog(n)\right).
\]
This is 
$\Otil\left(m^{1/2+1/(2i)+o(1)}\right)$
if we restart the structure every $T=m^{1-1/i}$ updates.
\end{proof}

\begin{cor}
[All-pairs effective resistance]\label{cor:ef fast}There is a dynamic
algorithm against an adaptive adversary that, given a weighted graph
$G$, handles the following operations in $m^{1/4}n^{1/2+o(1)}$
amortized update time:
\begin{itemize}
\item Insert or delete an edge from the graph, and
\item Given $s,t\in V(G)$, with probability at least $1-1/\poly(n)$, return
a $(1+\epsilon)$-approximation of the effective resistance between
$s$ and $t$.
\end{itemize}
\end{cor}

\begin{proof}
Chen et al.~\cite{chen2020fast} (from the proof of Theorem 8.1)
gives a dynamic algorithm $\A$ against an oblivious adversary with
the following guarantee. Given a weighted graph $G$ with $n$ vertices
and $m$ edges and any parameters $\beta$ and $d$, the algorithm
preprocesses $G$ in $t_{p}=O(\frac{m}{\beta^{5}}\cdot(\frac{\log n}{\epsilon})^{O(1)})$
time and then handles the following operations in $t_{u}=O(\beta^{-2d+3}(\frac{\log n}{\epsilon})^{O(d)})$
amortized time:
\begin{itemize}
\item edge insertions and deletions, and 
\item given $(s,t)$, update $(\texttt{src},\texttt{snk})\gets(s,t)$. 
\end{itemize}
That is, given $T$ operations above, the total update time is $t_{\total}=t_{p}+T\cdot t_{u}$.
The algorithm also supports queries for a $(1+\epsilon)$-approximation
of the $(\texttt{src},\texttt{snk})$-effective resistance in $t_{q}=O(n\beta^{d}(\frac{\log n}{\epsilon})^{O(d)})$
time. We set $d=\omega(1)$, say $d=O(\log\log n)$, and will set
$1/\beta=(n^{2}/m)^{\Theta(1/d)}$ meaning that $1/\beta=n^{o(1)}$.
This implies that $(\frac{\log n}{\epsilon})^{O(d)}/\beta^{O(1)}=n^{o(1)}$
(recall that we assume $\epsilon$ to be a constant), which will help
us simplifying several factors in the update time below.

By plugging $\A$ into \Cref{thm:basic}, we obtain an algorithm against
an adaptive adversary with amortized update time of 
\[
\Otil\left(\frac{t_{p}+T\cdot t_{u}}{\sqrt{T}}+t_{q}\right)=\left(\frac{m+T\cdot\beta^{-2d}}{\sqrt{T}}+n\beta^{d}\right)\cdot n^{o(1)}
\]
Now, to balance the first two terms in the bound, we set $T=m\beta^{2d}$.
This gives the bound of
\[
(\sqrt{m}\beta^{-d}+n\beta^{d})\cdot n^{o(1)}=m^{1/4}n^{1/2+o(1)}
\]
by setting $\beta^{d}=m^{1/4}/n^{1/2}$ to balance the terms. Indeed, $1/\beta=(n^{2}/m)^{\Theta(1/d)}$ as promised.
\end{proof}

\subsection{Speed up via Sparsification against an Adaptive Adversary}

Here, we show how to speed up the update time of the algorithms from
\Cref{sec:apply reduction}. The idea is simple. Given a dynamic graph
$G$ with $n$ vertices, we want to use a dynamic algorithm against
an adaptive adversary for maintaining a \emph{sparsifier} $H$ of
$G$ where $H$ preserves a certain structure of $G$ (e.g. cuts,
distances, or effective resistances) but $H$ contains at most $\Otil(n)$
edges, and then apply the algorithms \Cref{sec:apply reduction} on
$H$. So the final update time only depends on $n$ and not $m$. 

\paragraph{Preliminaries on graph sparsification.}

To formally use this idea, we recall  some definitions related
to sparsification of graphs. Given a (weighted) graph $G$, we say
that $H$ is an \emph{$\alpha$-spanner} of $G$ if $H$ is a subgraph
of $G$ and the distances between all pairs of vertices $s,t\in G$
are preserved with in a factor of $\alpha$, i.e. 
\[
\dist_{G}(u,v)\le\dist_{H}(u,v)\le\alpha\cdot\dist_{G}(u,v).
\]
We say that $H$ is an \emph{$\alpha$-cut sparsifier} of $G$ if,
for any cut $(S,V(G)\setminus S)$, we have that 
\[
\delta_{G}(S)\le\delta_{H}(S)\le\alpha\cdot\delta_{G}(S)
\]
where $\delta_{G}(S)$ and $\delta_{H}(S)$ denote the cut size of
$S$ in $G$ and $H$, respectively. Also, we say that $H$ is an
\emph{$\alpha$-spectral sparsifier} of $G$, if for any vector $x\in\mathbb{R}^{V}$,
we have that
\[
x^{\top}L_{G}x\le x^{\top}L_{H}x\le \alpha\cdot x^{\top}L_{G}x
\]
where $L_{G}$ and $L_{H}$ are the Laplacian matrices of $G$ and
$H$, respectively. 
\begin{fact}
Suppose that $H$ is an $\alpha$-spectral sparsifier of $G$.
\begin{itemize}
\item $H$ is an $\alpha$-cut sparsifier of $G$ (see e.g.~\cite{spielman2011spectral}). 
\item The effective resistance in $G$ between all pairs $(s,t)$ are approximately
preserved in $H$ up to a factor of $\alpha$. That is,
\[
R_{G}(u,v)\le R_{H}(u,v)\le\alpha\cdot R_{G}(u,v)
\]
for all $u,v\in V$ where $R_{G}(u,v)$ and $R_{H}(u,v)$ denote the
effective resistance between $u$ and $v$ in $G$ and $H$, respectively
(see \cite{durfee2019fully} Lemma 2.5).
\end{itemize}
\end{fact}

\paragraph{Using dynamic spanners as a blackbox.}

Bernstein et al.~\cite{bernstein2020fully} showed how to maintain
a $\polylog(n)$-spanner against an adaptive adversary efficiently.
\begin{theorem}
[\cite{bernstein2020fully} (Theorem 1.1)]\label{thm:spanner}There
is a randomized dynamic algorithm against an adaptive adversary that,
given an $n$-vertex graph $G$ undergoing edge insertions and deletions,
with high probability, explicitly maintains a $\polylog(n)$-spanner
$H$ of size $\Otil(n)$ using $\polylog(n)$ amortized update time.
\end{theorem}

With the above theorem, we can immediately speed up \Cref{cor:distance crude slow}
as follows:
\begin{cor}
[All-pairs distances]\label{cor:distance sparse}For any integer
$i\ge2$, there is a dynamic algorithm against an adaptive adversary
that, given a weighted graph $G$, handles the following operations
in $n^{1/2+1/(2i)+o(1)}$ amortized update time:
\begin{itemize}
\item Insert or delete an edge from the graph, and
\item Given $s,t\in V(G)$, with probability at least $1-1/\poly(n)$, return
a $O(\log^{3i+O(1)} n)$-approximation of the distance between $s$
and $t$.
\end{itemize}
\end{cor}

\begin{proof}
Using \Cref{thm:spanner}, we maintain a $\polylog(n)$-spanner $H$
of $G$ containing $\Otil(n)$ edges in $\polylog(n)$ amortized update
time. Since $H$ can be maintained against an adaptive adversary.
We think of $H$ as our input graph and pay an additional $\polylog(n)$
approximation factor and update time. The argument now proceeds exactly
in the same way as in the proof of \Cref{cor:distance crude slow}
except that now the input graph always contains at most $\Otil(n)$.
\end{proof}

\subsection{Current Limitation of Sparsification against an Adaptive Adversary}

Here, we discuss the current limitation of dynamic sparsifiers against
an adaptive adversary, which explains why we could not apply the same
idea to get $(1+\epsilon)$-approximation algorithms against adaptive
adversary.

\paragraph{Spanners.}

Observe that as long as we work on top of a $\polylog(n)$-spanner,
we will need to pay additional $\polylog(n)$ factor in the approximation
factor. We can hope to improve this factor because, against an oblivious
adversary, there actually exists an algorithm by Forster and Goranci
\cite{forster2019dynamic} for maintaining a $(2k-1)$-spanner of
size at most $\Otil(n^{1+1/k})$ edges using only $O(k\log^{2}n)$
amortized update time for any integer $k\ge1$. This approximation-size
trade-off is tight assuming the Erdos conjecture.

If there was a dynamic spanner algorithm with similar guarantees that works against an adaptive
adversary, we would be able to reduce the additional approximation
factor, due to sparsification, from $\polylog(n)$ to $2k-1$, while the update time would
be slightly slower because the sparsifiers have size $\Otil(n^{1+1/k})$
instead of $\Otil(n)$. Unfortunately, it is an open problem
if such a algorithm exists. This is the main reason why we could not
state the sparsified version of \Cref{cor:distance accurate slow}
where the approximation ratio remains $\log(n)\cdot\poly(\log\log n)$.
\begin{question}
Is there a dynamic algorithm against an adaptive adversary for maintaining
a $(2k-1)$-spanner of size $\Otil(n^{1+1/k})$ using $\polylog(n)$
update time?
\end{question}

\paragraph{Cut/spectral sparsifiers.}

The situation is similar for problems related to cuts and effective
resistance. Against an oblivious adversary, there exists a dynamic
algorithm by Abraham et al.~\cite{abraham2016fully} for maintaining
$(1+\epsilon)$-spectral sparsifier containing only $\Otil(n)$ edges
with $\polylog(n)$ amortized update time. If there was an algorithm
against an adaptive adversary with the same guarantees, then we could
immediately use it in the same manner as in \Cref{cor:distance sparse}
to speed up all the results from \Cref{cor:globalmincut} and \Cref{cor:ef fast}
by replacing $m$ by $n$ in their update time, while paying only
an extra $(1+\epsilon)$-approximation ratio in the query. Unfortunately,
whether such a dynamic algorithm for $(1+\epsilon)$-spectral sparsifier
exists still remains a fascinating open problem. 
\begin{question}
Is there a dynamic algorithm against an adaptive adversary for maintaining
a $(1+\epsilon)$-spectral sparsifier of size $\Otil(n)$ using $\polylog(n)$
update time?
\end{question}

The current start-of-the-art of dynamic cut/spectral sparsifiser algorithms
against an adaptive adversary still have large approximation ratio.
In particular, \cite{bernstein2020fully} show how to maintain a $\polylog(n)$-spectral
sparsifier in $\polylog(n)$ update time, and also a $O(k)$-cut sparsifiers
in $\Otil(n^{1/k})$ update time. We could apply these algorithms
to speed up Corollaries \ref{cor:globalmincut} and \ref{cor:ef fast}, but then
we must pay a large additional approximation factor. 

Fortunately, in \Cref{sec:sparsification period}, we are able to show
a way to work around this issue.

\subsection{Speed up via Sparsification against a Blinking  Adversary}

\label{sec:sparsification period}

In this section, we show that even if dynamic $(1+\epsilon)$-spectral
sparsifiers against an adaptive adversary are not known, we can still
apply the sparsification idea to the whole reduction. Below, (1) we
describe the sparsification lemma in \Cref{lem:sparsify from expander}
and discuss why we can view it as an algorithm for an intermediate
model between oblivious and adaptive adversaries which we call a \emph{blinking adversary}, defined in Section \ref{sec:blinking} and then (2) we apply \Cref{lem:sparsify from expander}
to speed up our applications.

The algorithm is based on dynamic expander decomposition. We recall
the definition of expanders here. Given a weighted graph $G=(V,E,w)$
and a vertex set $S\subseteq V$, the \emph{volume} of $S$ is $\vol_{G}(S)=\sum_{u\in S}\deg_{G}(u)$
where $\deg_{G}(u) = \sum_{(u,v)\in E} w(u,v)$ is the weighted
degree of $u$ in $G$. We say
that $G$ is a \emph{$\phi$-expander} if for any cut $(S,V\setminus S)$,
its conductance $\Phi(G)=\min_{(S,V\setminus S)}\frac{\delta_{G}(S)}{\min\{\vol_{G}(S),\vol_{G}(V\setminus S)\}}\ge\phi$.\textbf{ }

The following lemma says that for any graph $G=(V,E)$ with $n$ vertices
and $m$ edges, there exists a partition/decomposition of edges of
$G$ into $\phi$-expanders where each vertex appears in at most $O(\log^{3}n)$
expanders on average. Moreover, this decomposition can be maintained
against an adaptive adversary.

\begin{theorem}
[Dynamic Expander Decomposition  \cite{bernstein2020fully} (Theorem 4.3)]\label{thm:dynamic exp decomp}For
any $\phi=O(1/\log^{4}m)$ there exists a dynamic algorithm against
an adaptive adversary, that preprocesses a weighted graph $G$ with
$n$ vertices and $m$ edges in $O(\phi^{-1}m\log^{6}n)$ time. The
algorithm maintains with probability $1-1/\poly(n)$  a decomposition
of $G$ into $\phi$-expanders $G_{1},\dots,G_{z}$. The graphs $(G_{i})_{1\le i\le z}$ are edge disjoint and we
have $\sum_{i=1}^{z}|V(G_{i})|=O(n\log{}^{3}n)$. The algorithm
 supports
edge deletions and insertions in $O(\phi^{-2}\log^{7}n)$ amortized
time.
After each update,
the output consists of a list of changes to the decomposition. The
changes consist of (i) edge deletions or deletions of isolated vertices
to some graphs $G_{i}$ , (ii) removing some graphs $G_{i}$ from
the decomposition, and (iii) new graphs $G_{i}$ are added to the
decomposition.\footnote{Theorem 4.3 in \cite{bernstein2020fully} is stated only for unweighted
graphs. However, it is straightforward to extend the algorithm to
weighted graphs by grouping edges of $G$ by weights. As there are
$O(\log n)$ groups assuming that edge weights are at most $\poly(n)$,
this only add an extra factor of $O(\log n)$ to all the bounds.}
\end{theorem}

Actually the algorithm can be made deterministic when $\phi=O(1/2^{\log^{4/5}n})=1/n^{o(1)}$
by using the deterministic expander decomposition from \cite{chuzhoy2020deterministic}.
This would only add an extra $n^{o(1)}$ factor in the update time
of our applications, but in a first read, the reader may assume for
simplicity that \Cref{thm:dynamic exp decomp} is deterministic.

Given the dynamic expander decomposition from \Cref{thm:dynamic exp decomp},
we show that we can generate a $(1+\epsilon)$-spectral sparsifier
in $\Otil(n)$ time and even maintain it dynamically if the adversary
is oblivious to the sparsifier. We emphasize that the time is independent
of $m$ and depends only on $n$.
\begin{lem}
\label{lem:sparsify from expander}Let $\counter$ be the variable
that counts the total number of edge changes in all $\phi$-expanders
$G_{1},\dots,G_{z}$ of $G$ from \Cref{thm:dynamic exp decomp}. We
can extend the algorithm from \Cref{thm:dynamic exp decomp} so that
it handles the following additional operation:
\begin{itemize}
\item $\sparsify(t)$: return a graph $H$ and continue maintaining $H$
until $\counter$ has increased by $t$ using $\Otil(n+t)$ total
update time. During the sequence of edge updates before $\counter$
has increased by $t$, $H$ contains $\Otil(n+t)$ edges and there
are at most $\Otil(t)$ edge changes in $H$. If these edge updates
are fixed at the time $\sparsify(t)$ is called, then with high
probability, $H$ is a $(1+\epsilon)$-spectral sparsifier throughout
the update sequence. We emphasize that each call to $\sparsify(t)$
use fresh random bits to initialize and maintain $H$ in this call.
\end{itemize}
\end{lem}

The proof of \Cref{lem:sparsify from expander}  combines technical
ingredients from \cite{bernstein2020fully} and is given in \Cref{sec:proof sparsify}.

\paragraph{Discussion: Blinking Adversary.}
Assuming an underlying adaptive dynamic expander decomposition, we can think of the
 sparsifier of \Cref{lem:sparsify from expander} 
 as a sparsifier with 
 a fast refresh capability: Each
 call to {\em sparsify} refreshes the sparsifier and makes it accurate independently of the past updates. 

Suppose we always call $\sparsify(t)$ whenever $\counter$ increases
by $t$. Then we, in fact, maintain using $\Otil(n/t)$ amortized update time a $(1+\epsilon)$-spectral
sparsifier $H$ of $G$
against a blinking adversary as define in Section \ref{sec:blinking}. This is because, 
 \Cref{lem:sparsify from expander} guarantees accuracy of $H$  for a fixed sequence of updates (before $\counter$ increases by $t$), and each call to $\sparsify(t)$ uses
fresh random bits. So although the adversary observes $H$ following the last call to $\sparsify(t)$, this does not give it any useful information to fool the next call to  $\sparsify(t)$.

This model of an adversary lies between an oblivious adversary who
cannot observe the algorithm's answers at all and an adaptive adversary
who can observe the algorithm's answers after every update.
Interestingly, as we will see below, algorithms against
this intermediate model of adversary is enough for speeding up the
applications obtained by our generic reduction.

\begin{cor}
[Global minimum cuts (sparsified)]\label{cor:global mincut sparse}There
is a dynamic algorithm against an adaptive adversary that, given a
weighted graph $G$ undergoing edge insertions and deletions, with
probability $1-1/\poly(n)$, maintains a $(1+\epsilon)$-approximate
value of the global mincut in $\Otil(n^{3/4})$ amortized update time.
\end{cor}

\begin{proof}
First of all, we maintain the expander decomposition of $G$ in the
background using the algorithm from \Cref{thm:dynamic exp decomp}
with $\Otil(1)$ amortized update time when $\phi=\Theta(1/\log^{4}m)$.
Let $\counter$ be the variable that counts the total number of edge
changes in all $\phi$-expanders from the decomposition. 

We proceed in phases and restart the phase whenever
$\counter$ has increased by $T$ since the beginning of the phase (we will later set  $T=\sqrt{n}$).\footnote{It is possible that after one update to $G$, $\counter$ increases
by more than $T$. This means that within one update to $G$, there
can be more than one phase.} For each phase, our goal is to maintain a $(1+O(\epsilon))$-approximation
of the global minimum cut on $G$ against an \emph{oblivious} adversary
in time which is independent of $m$. Once this is done, we can apply
\Cref{thm:basic} to make it work against an adaptive adversary. Let
$c=\Otil(\sqrt{T})$ be the number of copies of the oblivious algorithms
in the reduction of \Cref{thm:basic}. 

To achieve this goal, at the beginning of the phase, we compute $H^{(1)},\dots,H^{(c)}$
using \Cref{lem:sparsify from expander} where each $H^{(j)}=\sparsify(T)$
are independently generated. For each instance $\A^{(j)}$ of the
oblivious algorithm for dynamic min cut by Thorup \cite{thorup2007fully},
we treat $H^{(j)}$ as its input graph. During the phase, $\counter$
can increase by at most $T$ and so by \Cref{lem:sparsify from expander}
there are $T_{H}=\Otil(T)$ updates to $H^{(j)}$. Therefore, during
the phase, the total update time of $\A^{(j)}$ for handling $T_{H}$
updates in $H$ is $t_{\total}=\Otil(|E(H^{(j)})|+T_{H}\sqrt{n})=\Otil(n+T_{H}\sqrt{n})$,
and the query time is $t_{q}=O(1)$ (recall that Thorup's algorithm maintains the solution explicitly). Since we can assume that the adversary for $\A^{(j)}$
is oblivious, \Cref{lem:sparsify from expander} guarantees that $H^{(j)}$
remains a $(1+\epsilon)$-spectral sparsifier of $G$ throughout the
phase. Therefore, each $\A^{(j)}$ indeed answers $(1+O(\epsilon))$-approximation
of the global min cut of $G$. 

By applying \Cref{thm:basic}, we can maintain $(1+O(\epsilon))$-approximate
solution of $G$ against an adaptive adversary where the total update
time is $\Otil(\sqrt{T_{H}}\cdot t_{\total}+T_{H}\cdot t_{q})=\Otil(\sqrt{T}\cdot n +  T^{3/2}\sqrt{n}+T\cdot t_{q})$
in this phase. The additional time for maintaining $H^{(1)},\dots,H^{(c)}$
is $c\times\Otil(n+T)=\Otil(\sqrt{T}\cdot (n+T))$. Therefore, the amortized
update time is 
\[
\Otil\left(\frac{\sqrt{T}\cdot n +  T^{3/2}\sqrt{n}}{T}+t_{q}\right)=\Otil(n^{3/4}),
\]
where the last equality holds if we choose $T=\sqrt{n}$.
\end{proof}

Via a similar modification of the proof of \Cref{cor:ef fast} by applying
\Cref{lem:sparsify from expander} and setting $T=n\beta^{2d}$, we obtain
the following. 
\begin{cor}
[All-pairs effective resistance (sparsified)]\label{cor:ef sparse}There
is a dynamic algorithm against an adaptive adversary that, given a
weighted graph $G$, handles the following operations in $n^{3/4+o(1)}$
amortized update time:
\begin{itemize}
\item Insert or delete an edge from the graph, and
\item Given $s,t\in V(G)$, with probability at least $1-1/\poly(n)$, return
a $(1+\epsilon)$-approximation of the distance between $s$ and $t$.
\end{itemize}
\end{cor}

\section{A Separation Result for a Search Problem}\label{sec:negative_search}

In this section, we describe and analyze a search problem that 
is much easier for dynamic algorithms against an oblivious adversary than for a dynamic algorithm against an adaptive adversary. For any constant $c>0$ and polynomial $P(n)\geq n$, we present a problem that can be solved by a dynamic algorithm against an oblivious adversary with amortized update time $P(n)$, however, (under some direct sum assumption) any dynamic  algorithm against an adaptive adversary requires $\Omega(P(n)\cdot n^{c})$  amortized update time.
That is, a dynamic algorithm against an adaptive adversary can save at most a constant multiplicative factor compared to an algorithm that solves the problem from scratch after each update, which can be done in time $P(n)\cdot n$. Notice that when we take a larger $c$ the separation is larger; however, this requires relying on a stronger direct-sum assumption.

We start with an informal description of the separation result.
We refer the reader to \cref{sec:intro-search} for a simpler and weaker separation result, providing some of the ideas in the separation proved in this section.
We fix a constant $c \in \NN$ and a function $H_n:\set{0,1}^{2n}\rightarrow \set{0,1}^n$ such that
$H_n(z)$ can be computed in time $P(n)$ for some polynomial $P(n)$.
We consider
the following dynamic problem  -- the adversary maintains two sets, a set $X$ of {\em excluded} prefixes, and a set $Y=\set{y_1,\dots,y_{n^c}}$ of suffixes. The set $X$ is initialized as the empty set and the initial set $Y$ is chosen by the adversary.
In each update, the adversary adds an element to $X$ and replaces one of the elements of $Y$. After each update, the algorithm has to output an element
$x$ \emph{not} in the set of excluded prefixes $X$ and 
the $n^c$ values  $H_n(x\circ y_1), \dots,H_n(x \circ y_{n^c})$, where $\circ$ denotes concatenation. 
An oblivious algorithm can pick $x$ at random in the preprocessing stage and compute $H_n(x\circ y_1), \dots,H_n(x \circ y_{n^c})$ for the initial set $Y$;
the oblivious algorithm will use the same $x$ throughout the entire run of the algorithm.
After each update the oblivious algorithm only needs to  compute $H_n(x \circ y)$ for the new $y$ that was inserted to $Y$, that is, the update time is $P(n)$.

For an adaptive algorithm life is much harder -- we consider an adaptive adversary that sees the element $x_{i-1}$ in output of the algorithm after the $(i-1)$-th update and adds $x_{i-1}$ to $X$. Thus, the algorithm has to use a new value $x_i$ after the $i$-th update.  Furthermore, the adaptive adversary in each update chooses some element $y \in Y$ and replaces it with a random element from $\set{0,1}^n$. 
As the algorithm does not know in the preprocessing time what the set $Y$ will be after $i>n^c$ updates, it cannot computes  the strings $H_n(x_i\circ y_1), \dots,H_n(x_i \circ y_{n^c})$ in the preprocessing stage. Thus, in the first  $2n^c$ updates, the algorithm has to compute the function $H_n$ on at least $n^{2c}$ values that it did not know in the preprocessing stage (and could not guessed them  with non-negligible probability). Assuming that $H_n$ has the direct-sum property  
for $n^{2c}$ elements implies that during the first $2n^c$ updates the work of the adversary is $\Omega(n^{2c} \cdot P(n))$, thus the amortized update time is $O(n^c P(n))$.
This lower bound holds even if the adaptive algorithm can compute as many as $2^{0.5n}$ values of $H_n$ in the preprocessing stage (as it would not guess with non-negligible probability what elements will be inserted to $Y$). We formalize these ideas in the rest of this section.

\begin{definition}[The List of Outputs Problem]
Fix a sequence of functions $\set{H_n:\set{0,1}^{2n}\rightarrow \set{0,1}^n}_{n \in \NN}$ and a constant $c>0$. 
In the dynamic setting of the List of Outputs Problem for strings of length $n$, a set $Y$ initially contains $n^c$ elements from $\set{0,1}^n$, a set $X \subset \set {0,1}^n $ is an arbitrary set such that $|X|\leq 2^{0.5n}$ (e.g., $X=\emptyset$), and the sets are updated during the lifetime of the algorithm, where every update replaces one of the elements in $Y$ and adds one element to $X$. 
After each update, the algorithm has to return an output  that consists of an element $x \notin X$ and a list of $|Y|$ values -- $H_n(x \circ y)$ for  every $y \in Y$ (for the current sets $X,Y$).
We assume that the number of updates is not too large and throughout the execution $|X| < 2^{0.5n}$, thus, there are always many legal outputs.
\end{definition}

Assume that computing $H_n(x \circ y)$ takes $P(n)$ time for some polynomial $P(n)$ (e.g., $P(n)=O(n)$). 
We next describe a simple dynamic algorithm against an oblivious adversary for the list of outputs problem. In the preprocessing stage, the algorithm picks uniformly at random an element $x \in \left(\{0,1\}^n\setminus X\right)$ and for that $x$ it computes $H_n(x \circ y)$ for all $y \in Y$. The complexity of the preprocessing is $n^c\cdot P(n)$. With high probability, none of the  futures update by an oblivious adversary will ``hit'' the element $x$ chosen by the algorithm (as we assume that the number of updates is not to large and $|X| \leq 2^{0.5 n}$), and therefore, after every update the algorithm  only needs to recompute $H_n(x \circ y)$ for the one ``new'' point $y$ that was inserted, that is, the complexity of each update is $P(n)$.

We now want to show that in the adaptive setting the amortized update time is nearly as high as computing one instance of the problem from scratch. As the list of outputs problem requires computing $H_n$ on $n^c$ inputs, we need to assume that the computation of many instances of $H_n$ cannot be done significantly more efficiently than computing each instance independently, i.e., we need to assume that $H_n$ has a direct-sum property. We formalized this requirement in the following definitions.

\begin{definition}
 We say that an algorithm $A$  computes $\set{H_n}_{n \in \NN}$ on $\ell(n)$ inputs if for every large enough $n$ and every inputs $z_1,\dots,z_{\ell(n)}\in \set{0,1}^{2n}$
 with probability at least $3/4$ the algorithm $A$ on inputs $z_1,\dots,z_{\ell(n)}$ 
 outputs  $H_n(z_1),\dots,H_n(z_{\ell(n)})$, where the probability is over the randomness of $A$.
\end{definition}

\begin{definition}
We say that a distribution $\calD$ on $(\set{0,1}^{2n})^{\ell(n)}$ is \emph{well-spread} if for every set $S \subset\set{0,1}^{2n}$
of size at most $2^{0.75n}$ if we sample a sequence 
$z_1,\dots,z_{\ell(n)}$ according to $\calD$, then the probability that 
$$\left|\set{z_1,\dots,z_{\ell(n)}} \cap S\right| < \ell(n)/2 $$
is at least  $7/8$. 
\end{definition}

Informally, a distribution is well-spread if even if in a preprocessing stage an adversary computes the value of $H_n$ on many points,
then a sequence of $\ell(n)$ elements sampled according to $\calD$ contains with probability at least $7/8$ at least $\ell(n)/2$ ``fresh'' elements.

\begin{definition}

A sequence of functions $\mathcal{H} = \set{H_n:\set{0,1}^{2n}\rightarrow \set{0,1}^n}_{n \in \NN}$ has the $\ell(n)$-direct sum property if:
\begin{itemize}
    \item 
    $H$ is moderately easy to compute, that is, there exists some polynomial $P(n)$ and an algorithm $A$ such that for every $z\in \set{0,1}^{2n}$, the algorithm $A$ computes $H_{n}(z)$ in at most $P(n)$ steps.
    \item 
    There is a constant $0 < \gamma \leq 1$ such that for every algorithm $A$ that computes $\set{H_n}_{n \in \NN}$ on $\ell(n)$  inputs, for every large enough $n$, for every well-spread distribution $\calD$ on $\ell(n)$ inputs, 
    the run time of $A(z_1,\dots,z_{\ell(n)})$ is at least $\gamma\cdot \ell(n) \cdot P(n)$ with probability at least $3/4$,
    where the probability is over the choice of $z_1,\dots,z_{\ell(n)}$ according to $\calD$ and over the randomness of $A$. The above should hold even  if for some set $S$ of size at most $2^{0.75n}$ the algorithm $A$ has computed 
    $H_n(z)$ for every $z \in S$ 
    prior to the choice of $z_1,\dots,z_{\ell(n)}$.
\end{itemize}
\end{definition}

In \cref{rem:RO} we discuss how to implement such sequences of functions $\mathcal{H}$.
\begin{remark}\label{rem:output}
We can make an assumption that writing $n$ bit outputs require time $\Omega(n)$. Under this ``physical'' assumption we can implement $H_n$ as the function that takes a $2n$ bit strings and writes them. Under this assumption we get the desired separation between oblivious algorithms and adaptive algorithms (even with the simplified version of the problem described in \cref{sec:intro-search}). However, this is a physical assumption on the process of writing outputs and it is not clear that this assumption is true (e.g., to write a string $x\circ y$ one can give pointers to the locations of the input strings $x$ and $y$). Thus, we prefer computational assumptions on the run time of computing $H_n$ on many inputs for some ``hard'' function $H_n$.
\end{remark}

We next show that a dynamic algorithm against an adaptive adversary has to have high amortized update time
even if the algorithm has a huge preprocessing time. The adaptive adversary 
will add each element $x$ that the dynamic algorithm outputs to the list of excluded points $X$. 
Intuitively, this makes the computation spent in previous rounds useless. 
Furthermore, consider a dynamic algorithm that invests computation time of $2^{o(n)}$ in the preprocessing stage with the goal of reducing update time. E.g., the algorithm may prepare a very long list of pairs $x,y$’s of size at most $2^{o(n)}$ for which it computes $H_n(x\circ y)$.
This algorithm would still need to invest a significant amount of time when updates happen, because of the changes made to the set $Y$. This intuition is formalized in the following lemma.

\begin{lemma}
Assume that ${\cal H} = \set{H_n:\set{0,1}^{2n}\rightarrow \set{0,1}^n}_{n \in \NN}$ has the $O(n^{2c})$-direct sum property.
There exists an adaptive adversary such that every dynamic algorithm solving the the list of outputs problem on a sequence of at most $2^{0.25n}$ updates
using at most $2^{0.5n} $ preprocessing time has at least $\Omega(n^{c} \cdot P(n))$ amortized update time.
\end{lemma}
\begin{proof}
Fix any dynamic algorithm $A$.
We consider an adaptive adversary that in each round $i>1$ (i) receives the output of the dynamic algorithm from round $i-1$, namely, an element  $x$ and a list of $|Y|$ outputs, (ii)  updates the set of excluded points $X$ by adding $x$ to it, and (iii) replaces the element $y\in Y$ that was earliest  added to $Y$ with a uniformly distributed element $y'\in_R \set{0,1}^n$
(in the first $n^c$ updates the adversary chooses $y$ from the original set $Y$ in an arbitrary order). 

We consider a sequence of $N \leq  2^{0.25 n}$ updates and divide the updates to $N/n^{c}$ blocks of $n^c$ updates each. We assume that the prepossessing time and the total work in each block is at most $2^{0.5n}$ (otherwise the lemma follows trivially), and prove that the total work time of the algorithm in each block is $\Omega(n^{2c} P(n))$. By our assumptions, algorithm $A$  prior to block $i$ has computed the value $H_n(z)$ on at most $2^{0.75n}$ elements.

Consider the updates in the $i$-th block. Let 
$x_1,\dots,x_{n^c}$ be the elements $x$ in the output of the algorithm and $y_1 ,\dots,y_{n^c}$ be the random elements inserted into $Y$ by the adversary. 
The uniform choice of $y_1 ,\dots,y_{n^c}$ and the choices of $x_1,\dots,x_{n^c}$ by the algorithm $A$ induce a probability distribution $\calD$ on the $O(n^{2c})$ points $(x_j \circ y_i)_{1 \leq j \leq n^c, 1 \leq i \leq j}$.
For every $A \subset \set{0,1}^{2n}$ of size at most $2^{0.75n}$  the probability that for a uniformly distributed $y \in \set{0,1}^n$ the set $A$ contains an element $x \circ y$ for some $x \in \set{0,1}^n$ is exponentially small. Thus, no matter how $A$  chooses $x_1,\dots,x_{n^c}$, the distribution $\calD$ is well spread.

Thus, during the $i$-th block, the algorithm has to compute $H_n(x_j \circ y_i)$ for every $1 \leq j \leq n^c$ and every $1 \leq i \leq j$; i.e., it has to compute $O(n^{2c})$ values of $H_n$.
Since ${\cal H}$ has the $O(n^{2c})$-direct sum property, computing these values will require expected time $\Omega(n^{2c} P(n))$ and, amortized over the $n^c$ updates in the block, it will require $\Omega(n^c P(n))$ amortized update time.
\end{proof}

\begin{remark}
\label{rem:RO}
We do not have a function that provably has
 the direct sum property. However, we can construct such a function from a random oracle (see below).
Thus, heuristically, a cryptographic hash function satisfies the direct-sum property. This heuristic assumption about the direct-sum property of cryptographic hash functions is used in the block-chain of the bitcoin~\cite{satoshi}. Specifically, in the proof of work in~\cite{satoshi} a miner has to compute hash values of many inputs until the hash value satisfies some condition; it is assumed that computing the hash values of many inputs can only be done by computing each value independently.

Technically, to implement a function $H_{n}: \set{0,1}^{2n}\rightarrow \set{0,1}$ using a random oracle $\RO$,
we define $b(z)$, for a $z \in \set{0,1}^{2n}$, as the  number whose binary representation is $1\circ z$ and $H_{n}(z)=(\RO[n \cdot b(z)+1],\RO[n \cdot b(z)+2],\dots,\RO[n \cdot b(z)+n])$. Thus, computing $H_{n}(z)$ will require $n$ time units, as any algorithm computing $H_n(z)$ has to read $n$ locations in the oracle.  Furthermore, the bits from the random oracle used in the computation of $H_n$ on different inputs are disjoint, thus, $H_n$ has the direct-sum property.
To construct a function $H'_n$ that can be computed in time $\theta(P(n))$ for some polynomial $P(n) >n$ we, for example,  define
$$H'_n(z)= \left(\oplus_{j=1}^{P(n)/n} \RO[n \cdot p(n) \cdot b(z)+j], \dots, 
\oplus_{j=P(n)-P(n)/n+1}^{P(n)} \RO[n \cdot p(n) \cdot b(z)+j]\right)$$
(where $\oplus$ is the exclusive-or operation),
that is, each of the $n$ bits in the output of $H'(n)$ requires reading $P(n)/n$ bits from the oracle.
\end{remark}

\section{Impossibility Result for an Estimation Problem}\label{sec:negative_est}

In Section~\ref{sec:negative_search} we presented a negative result for a {\em search} problem. In this section we present a negative result for an {\em estimation} problem. 

\subsection{Preliminaries from adaptive data analysis}
We first introduce additional preliminaries from the literature on adaptive data analysis \cite{dwork2015preserving}.
A \emph{statistical query} over a domain $X$ is specified by a predicate $q:X\rightarrow\{0,1\}$. The value of a query $q$ on a distribution $\mu$ over $X$ is $q(\mu)=\E_{x\sim\mu}[q(x)]$.
Given a database $S\in X^n$ and a query $q$, we denote the empirical average of $q$ on $S$ as $q(S)=\frac{1}{n}\sum_{x\in S}{q(x)}$.

In the adaptive data analysis (ADA) problem, the goal is to design a \emph{mechanism} $\MMM$ that answers queries on $\mu$ using only i.i.d.\ samples from $\mu$. Our focus is the case where the queries are chosen adaptively and adversarially.
Specifically, $\MMM$ is a stateful algorithm that holds a collection of samples $(x_1,\dots,x_n)$. In each round, $\MMM$ receives a statistical query $q_i$ as input and returns an answer $z_i$. We require that when $x_1,\dots,x_n$ are independent samples from $\mu$, then the answer $z_i$ is close to $q_i(\mu)$. Moreover we require that this condition holds for every query in an adaptively chosen sequence $q_1,\dots, q_{\ell}$. Formally, we define an accuracy game $\accgame_{n, \ell,\MMM,\mathbb{A}}$ between a mechanism $\MMM$ and a stateful \emph{adversary} $\mathbb{A}$ in Algorithm~\ref{fig:accgame1}.

\begin{algorithm*}[t!]
\caption{The Accuracy Game $\accgame_{n, \ell,\MMM(S),\mathbb{A}}$.}\label{fig:accgame1}

\begin{enumerate}

\item For $i=1$ to $\ell$, 

\begin{enumerate}
	\item The adversary $\mathbb{A}$ chooses a statistical query $q_i$.
	\item The mechanism $\MMM$ receives $q_i$ and outputs an answer $z_i$.
	\item The adversary $\mathbb{A}$ gets $z_i$.
\end{enumerate}
\item Output the transcript $(q_1,z_1,\dots,q_{\ell},z_{\ell})$.
\end{enumerate}
\end{algorithm*}

\begin{definition}\label{def:accuratemechanism}
A mechanism $\MMM$ is \emph{$(\alpha,\beta)$-statistically-accurate for $\ell$ adaptively chosen statistical queries given $n$ samples} if for every adversary $\mathbb{A}$ and every distribution $\mu$,
\begin{equation}
\Pr_{\substack{S\sim\mu^n\\ \accgame_{n, \ell,\MMM(S), \mathbb{A}}}}\Big[
\exists i\in[\ell] \text{ s.t.\ }
\left| q_i(\mu)-z_i \right|>\alpha
\Big]\leq\beta.
\label{eq:accuratemechanism}
\end{equation}
\end{definition}

Hardt, Steinke and Ullman \cite{HardtU14,SteinkeU15} showed that there is no computationally efficient mechanism that given $n$ samples is accurate on $\gtrsim n^2$ adaptively chosen queries. Formally, 

\begin{theorem}[\cite{HardtU14,SteinkeU15}]\label{thm:SU15}
There is a constant $\lambda>1$ such that the following holds. Assuming one-way functions exist, there is no computationally efficient mechanism $\MMM$ that is $(0.49,0.5)$-statistically-accurate for $\ell=\lambda\cdot n^2$ adaptively chosen queries given $n$ samples in $\{0,1\}^n$. %
Formally, for every polynomial time mechanism $\MMM$ there exists an adversary $\mathbb{A}$ and a distribution $\mu$ over $\{0,1\}^n$ such that
$$
\Pr_{
\substack{S\sim\mu^n\\ \accgame_{n, \ell,\MMM(S), \mathbb{A}}}}\Big[
\exists i\in[\ell] \text{ s.t.\ }
\left| q_i(\mu)-z_i \right|\geq0.49
\Big]\geq\frac{1}{2}.
$$
Furthermore: 
\begin{enumerate}
    \item The adversary $\mathbb{A}$ runs in $O(n^2)$-time per query.
    \item Every query $q_i$ that $\mathbb{A}$ poses can be described using $O(n^2)$ bits and can be evaluated in time $O(n^2)$ on any point $x\in\{0,1\}^n$.
    \item The distribution $\mu$ is uniform on a subset $Y\subseteq\{0,1\}^n$ of size $|Y|=O(n)$.
\end{enumerate}
\end{theorem}

\begin{remark}
The statement in \cite{HardtU14,SteinkeU15} is more general in the sense that the representation length of the samples (call id $d$) does not have to be equal to $n$ as in Theorem~\ref{thm:SU15}. The representation length $d$ can be set to $d=n^\gamma$ for any constant $\gamma>0$ without effecting the statement. In fact, assuming sub-exponentially secure one-way functions exist, we can even set $d=\polylog(n)$. On the other hand, if we assume that $d=\Omega(n^2)$, then the Theorem~\ref{thm:SU15} holds {\em unconditionally}. We chose to fix $d=n$ for simplicity.
\end{remark}

\subsection{Boxes schemes}

The negative result we present in this section assumes the existence of a primitive, formally defined next, that allows a {\em dealer} with $m$ inputs to encode each input in a ``box'' such that: 
\begin{enumerate}
    \item Retrieving the content of a box (i.e., the corresponding input) can be done in time $T$.
    \item Any $bT$-time algorithm $\BBB$ that operates on the boxes cannot learn information on the content of more than $O(b)$ of the boxes. We formalize this by comparing the outcome distribution of algorithm $\BBB$ to an ``ideal'' algorithm $\SSS$ that does not get the boxes as input and instead gets oracle access to $O(b)$ of the inputs. 
\end{enumerate}

The intuitive motivation for this primitive is that given an encryption key $k$ and a ciphertext $c$ (obtained by encrypting a plaintext $x$ using $k$), it should be hard to extract meaningful information on $x$ using significantly less runtime than the runtime needed to decrypt $c$. Formally,

\begin{definition}
Fix $x_1,\dots,x_m$. We write $\SSS^{\Oracle_{b}(x_1,\dots,x_m)}$ to denote an algorithm $\SSS$ with oracle access, limited to $b$ requests, where in every request the algorithm specifies an index $i\in[m]$ and obtains $x_i$.
\end{definition}

\begin{definition}\label{def:boxes}
Let $n\in\N$ be a security parameter, and  let $T,m,b\in\N$ be functions of $n$. A {\em $(T,m,b)$-boxes-scheme} is a tuple of algorithms $(\Gen,\Enc,\Dec)$ with the following properties.
\begin{enumerate}
    \item[{\bf (1)}] {\bf Syntax:} Algorithm $\Gen$ is a randomized algorithm that takes input $1^n$ and outputs a key $k\in\{0,1\}^n$. 
    Algorithm $\Enc$ is a randomized algorithm that takes a key $k\in\{0,1\}^n$ and an input $x\in \{0,1\}^n$ and returns an element $c\in\{0,1\}^*$. Algorithm $\Dec$ takes a key $k\in\{0,1\}^n$ and an element $\{0,1\}^*$ and returns an element in $\{0,1\}^n$. 
    \item[{\bf (2)}] {\bf Correctness:} For every $k\in\{0,1\}^n$ in the support of $\Gen(1^n)$ and for every $x\in \{0,1\}^n$, it holds that $\Dec(k,\Enc(k,x))=x$.
    \item[{\bf (3)}] {\bf Encoding and decoding time:} Algorithms $\Gen$ and $\Enc$ run in time $\poly(n)$. The running time of $\Dec$ on any key $k\in\{0,1\}^n$ and element $c\in\{0,1\}^*$ is at most $T(n)$.
    \item[{\bf (4)}] {\bf Vector decoding is hard:} For every $b\cdot T(n)$-time algorithm $\BBB$ there exists an ``ideal'' functionality $\SSS$ (called {\em simulator}) such that for every $\frac{T(n)}{4}$-time algorithm $\DDD$ (called {\em distinguisher}) there is an $n_0\in\N$ such that for every $n\geq n_0$ and every 
    $\vec{x}=(x_1,\dots,x_m)\in\left(\{0,1\}^n\right)^m$ %
    it holds that
    $$
    \left|
    \Pr\left[ \DDD\Big(\vec{x}\,,\,\BBB(k,\vec{c})\Big)=1  \right]
    -
    \Pr\left[ \DDD\left(\vec{x}\,,\,\SSS^{\Oracle_{2b}(\vec{x})}\right)=1  \right]
    \right|\leq\frac{1}{10}.
    $$
    The left probability is over sampling $k\leftarrow\Gen(1^n)$, encoding each $c_i\leftarrow\Enc(k,x_i)$, and over the randomness of $\BBB$ and $\DDD$. The right probability is over the randomness of $\SSS$ and $\DDD$.
\end{enumerate}
\end{definition}

\begin{remark}
Note that the runtime of $\Gen$ and $\Enc$ is required to be polynomial in $n$, but not necessarily at most $T(n)$. Also note that we did not restrict the runtime  of the simulator $\SSS$, as this is not needed for our application. In other settings, it might make sense to restrict the simulator's runtime or to further restrict the  runtime of $\Gen$ and $\Enc$. In addition, for other applications, one might require the advantage of the distinguisher to be negligible in $n$, rather than a constant. 
\end{remark}

\begin{remark}
We restrict the runtime of the distinguisher $\DDD$ to be strictly smaller than $T(n)$. Otherwise, $\DDD$ could open boxes by itself in order to distinguish between $\BBB$ and $\SSS$. The intuition here is that for an algorithm $\BBB$ to break the scheme we want it to be able to extract more information from the boxes than it should, and  present this information to $\DDD$ in an accessible way. If $\DDD$ needs to ``work  hard'' in order to make use of the  output of $\BBB$, then $\BBB$ did not really ``extract'' the information from the boxes. %
\end{remark}

\begin{remark}
Observe that if a 
$(T,m,b)$-boxes-scheme
 $(\Gen,\Enc,\Dec)$ satisfies Definition~\ref{def:boxes} then, in particular, boxes generated by this scheme cannot be opened in time $\frac{T(n)}{4}$ (provided that $b\ll m$).  To see this, assume towards contradiction that boxes can be opened in time $\frac{T(n)}{4}$.
Now let $b<m/100$,  and let $\BBB$ be the algorithm that simply outputs all its inputs $(k,\vec{c})$.\footnote{We assume that the time this takes is significantly smaller compared to $\BBB$'s  allowed runtime of $b\cdot T(n)$. Essentially this means that $b\cdot T(n) \gg mn$, assuming that encodings $c_i$ are represented using $O(n)$ bits.}
Then a distinguisher $\DDD$ (who knows all the the $x_i$'s) can pick a random index $i$, open the $i$th box in time $T(n)/4$ and check if it got $x_i$. This test will pass w.p.\ 1 in the real world, but with probability at most 1/50
in the ideal world since the simulator $\SSS$ cannot see more than $m/50$ inputs. So if boxes can be opened in time $T(n)/4$ then there is a distinguisher that contradicts the assumption that $(\Gen,\Enc,\Dec)$ satisfies the definition. %
\end{remark}

\begin{remark}
We do not have a construction for which we can prove that it satisfies Definition~\ref{def:boxes}. However, this can be achieved in the random oracle model as follows. We can think of a random oracle as a random function $\RO:\{0,1\}^n\rightarrow\{0,1\}^n$. To encode an input point $x\in\{0,1\}^n$, algorithm $\Enc$ picks a random point $p\in\{0,1\}^n$ and identifies $k_p=\RO(\RO(\RO(\dots \RO(p))))$ by applying a composition of $T(n)$ applications of $\RO$ to $p$. The encoding of $x$ is defined as $(p, k_p\oplus x)$. Given this encoding, one can recover $x$ in time $T(n)$ by recovering $k_p=\RO(\RO(\RO(\dots \RO(p))))$. However, any algorithm that runs in time smaller than $T(n)$ learns nothing about $x$.
We stress that when encoding a collection of plaintext inputs $x_1,\dots,x_m$ then a different point $p$ is picked for any $x$. 

Note that with this construction the runtime of the encoder is as big as the runtime of the decoder, which suffices for the results we present in this section. In general, however, one might be interested in a construction where the runtime of $\Enc$ is significantly smaller. This too can be achieved in the random oracle model as follows. To encode an input point $x\in\{0,1\}^n$, we pick a random point $p\in\{0,1\}^n$ and pick a random point $k_p\in[p,p+T(n)]$. The encoding of $x$ is defined as $\left(p, \RO(k_p), x\oplus \RO(k_p+1) \right)$. Given this encoding, one can recover $x$ in time $T(n)$ by searching for an element $k$ in the range $[p,p+T(n)]$ such that $\RO(k)=\RO(k_p)$, and then recovering $x$ using $\RO(p+1)$. However, any algorithm with runtime $\ll T(n)$ has a small probability of learning $\RO(k_p+1)$, and without this it learns nothing on the input $x$. %
\end{remark}

\subsection{The average-of-boxes problem}

\begin{definition}[The Average-of-Boxes Problem]
Fix a $(T,m,b)$-boxes-scheme  $(\Gen,\Enc,\Dec)$. 
The algorithm gets 
a key $k\in\{0,1\}^n$, and a collection of $m$ bit strings (``boxes'') $c_1,\dots,c_m\in\{0,1\}^*$.
In each round the adversary picks a
 function $f:\{0,1\}^n\rightarrow\{0,1\}$, described using $O(n^2)$ bits,
 which can be 
  evaluated  on any (plaintext) input $x\in\{0,1\}^n$ in $O(n^2)$ time. 
  The algorithm has to respond with (an approximation for) $\frac{1}{m}\sum_{i=1}^m f(\Dec(c_i))$. 
  The sequence consists of $\ell$ updates and responses.
  We say that an algorithm $\MMM$ solves the average-of-boxes problem with parameters $(\alpha,\beta)$ if it guarantees that with probability at least $(1-\beta)$ all of its answers are accurate to within a $(1\pm\alpha)$ multiplicative error.
\end{definition}

\begin{remark}
Observe that the average-of-boxes problem is an estimation problem.
\end{remark}
We first describe a simple dynamic algorithm against an oblivious adversary for the average-of-boxes problem. In the preprocessing stage, the algorithm uses $\Dec$ to open a random subset of $w=O\left(\frac{1}{\alpha^2}\log(\frac{\ell}{\beta})\right)$ of the boxes $c_1,\dots,c_n$ to obtain $w$ (plaintext) inputs $x_{i_1},\dots,x_{i_w}$. Then, on every round, the algorithm returns the average of the current function $f$ on the $x_{i_j}$'s. By the Chernoff bound, all of the $\ell$ answers given by the algorithm are accurate w.r.t.\ their corresponding functions. The total runtime of this algorithm (for preprocessing and the $\ell$ updates) is $O(w\cdot T(n) + w\cdot \ell\cdot n^2)=\tilde{O}\left( T(n) +  \ell\cdot n^2\right),$ where $T(n)$ is the runtime of algorithm $\Dec$. In our application, we will set $T(n)\gg\ell n^2$, and hence the runtime of the oblivious algorithm is $\tilde{O}\left(T(n)\right)$.

We next show that in the adversarial setting, the average-of-boxes problem requires significantly more time. To this end, assume that we have an algorithm $\MMM$ that solves the average-of-boxes problem in the adversarial setting. 
If $\MMM$ is fast then we can 
apply the guarantees of the underlying boxes-scheme to 
$\MMM$ (with $\MMM$ playing the role of 
$\BBB$ in part (4) of Definition \ref{def:boxes}) to argue that 
 there is a simulator $\SSS$ that
answer the same queries as 
$\MMM$  using only
 a small number of plaintext inputs $x_i$.
As $\MMM$ solves the average-of-boxes problem, this simulator  accurately answers a sequence of adaptively chosen queries w.r.t.\ the collection of {\em all} inputs using only a very small number of samples from them. This can be formalized to show that the simulator is essentially solving the ADA problem while using a very small number of samples, and this should not be possible by the negative results on the ADA problem. The difficulty with this argument is that algorithm $\MMM$ is {\em interactive} (it answers queries) and our notion of boxes-scheme is defined for non-interactive algorithms. In the following lemma, we overcome this difficulty by applying the guarantees of the boxes-scheme not to $\MMM$ directly, but rather to the composition of $\MMM$ with the (adversarial) algorithm that prepares the queries.

\begin{lemma}\label{lem:main_boxes}
There exists a constant $\gamma>0$ s.t.\ the following holds. 
Assume the
existence of one-way functions and $(T,m,b)$-boxes-scheme $(\Gen,\Enc,\Dec)$ with security parameter $n$ for $m=n$ and $b=n/4000$ and $T(n)>\lambda\cdot n^5$. Then, every algorithm for the average-of-boxes problem must have total runtime $\Omega\left(n\cdot T(n)\right)$ over $\ell=\Theta(n^2)$ rounds.
\end{lemma}

For example, setting $T(n)=n^6$ we get that the total runtime of the oblivious algorithm is $O(n^6)$ over $\ell=n^2$ updates, while the total runtime in the adversarial setting is $\Omega(n^7)$. This translates to an average runtime per query of $O(n^3)$ in the oblivious setting vs.\ $\Omega(n^4)$ in the adversarial setting.
The second part of Theorem 
\ref{thm:negative_intro}
now follows.

\begin{remark}
We remark that the difference in the average time could be pushed to $\tilde{O}(1)$ vs.\ $\Omega(n)$. The easiest (but artificial) way to achieve this is to ``pad'' the time steps with ``dummy'' time steps during which ``nothing happens''. This does not effect the total runtime, but reduces the average runtime to $\tilde{O}(1)$ vs.\ $\Omega(n)$. In the case of the average-of-boxes problem, this could be done in a less artificial way, where in every round we update only a small fraction of the description of the current function.
\end{remark}

\begin{proof}[Proof of Lemma~\ref{lem:main_boxes}]
Assume towards a contradiction that there exists an algorithm $\MMM$ that solves the average-of-boxes problem with $(\Gen,\Enc,\Dec)$ as specified in the lemma, with accuracy $\alpha=0.1$, and success probability $\beta=3/4$, and total runtime at most $\frac{1}{2}\cdot \sqrt{\ell} \cdot T(n)$. 
So algorithm $\MMM$ obtains an input of the form $(k,\vec{c})$, representing a key and a collection of boxes, and then answers $\ell$ adaptively chosen queries. We use $\MMM$ to construct an algorithm for the ADA problem. 
A small mismatch is that $\MMM$ obtains encoded inputs, while algorithms for the ADA problem run on plaintext inputs. To bridge this mismatch, consider a preprocessing algorithm $\WWW$ that takes an input dataset $X\in(\{0,1\}^n)^n$ containing $n$ plaintext inputs in $\{0,1\}^n$, then samples a key $k\leftarrow\Gen(1^n)$, encodes  $c_i\leftarrow\Enc(k,x_i)$, and executes $\MMM$ on $(k,c_1,\dots,c_n)=(k,\vec{c})$. (So algorithm $\WWW$ only preprocesses the plaintext inputs, and runs $\MMM$ on the encoded inputs.) %
We got another (interactive) algorithm, namely
 $\MMM(\WWW(X))$, %
 that takes a dataset $X$ of size $n$ and then answers $\ell$ adaptively chosen queries on it.

 By Theorem~\ref{thm:SU15}, assuming that $\ell= \Omega(n^2)$, there exists a distribution $\mu$ (uniform on a set of $O(n)$ elements) and an adversary $\mathbb{A}$ (that runs in total time $O(\ell n^2)$ for $\ell$ rounds) such that %
$$
\Pr_{
\substack{\vec{x}\sim\mu^n\\ \accgame_{n, \ell,\MMM(\WWW(\vec{x})), \mathbb{A}}}}\Big[
\exists i\in[\ell] \text{ s.t.\ }
\left| q_i(\mu)-z_i \right|\geq0.49
\Big]\geq\frac{1}{2}.
$$

Equivalently, explicitly writing the operation of the preprocessing algorithm $\WWW$, we have
$$
\Pr_{
\substack{\vec{x}\sim\mu^n\\ 
\Gen,\Enc\\ 
\accgame_{n, \ell,\MMM(k,\vec{c}), \mathbb{A}}}}\Big[
\exists i\in[\ell] \text{ s.t.\ }
\left| q_i(\mu)-z_i \right|\geq0.49
\Big]\geq\frac{1}{2}.
$$
On the other hand, by the assumption that $\MMM$ solves the average-of-boxes problem, its answers must be accurate w.r.t.\ the empirical average of the queries on $\vec{x}$ with probability at least $3/4$. Hence, by a union bound,
$$
\Pr_{
\substack{\vec{x}\sim\mu^n\\ 
\Gen,\Enc\\ 
\accgame_{n, \ell,\MMM(k,\vec{c}), \mathbb{A}}}}\Big[
\exists i\in[\ell] \text{ s.t.\ }
\left| q_i(\mu)-z_i \right|\geq0.49
\quad
\text{and}
\quad
\forall i\in[\ell] \text{ we have }
\left| q_i(\vec{x})-z_i \right|\leq0.1
\Big]\geq\frac{1}{4}.
$$
In particular, by the triangle inequality, we get that
$$
\Pr_{
\substack{\vec{x}\sim\mu^n\\ 
\Gen,\Enc\\ 
\accgame_{n, \ell,\MMM(k,\vec{c}), \mathbb{A}}}}\Big[
\exists i\in[\ell] \text{ s.t.\ }
\left| q_i(\mu)-q_i(\vec{x}) \right|\geq0.39
\Big]\geq\frac{1}{4}.
$$
Let us denote $\BBB(k,\vec{c})=\accgame_{n, \ell,\MMM(k,\vec{c}), \mathbb{A}}$, and observe that $\BBB$ runs in time at most $\frac{1}{2}\cdot \sqrt{\ell} \cdot T(n) + O(\ell n^2)$, where $O(\ell n^2)$ bounds the runtime of algorithm $\mathbb{A}$. We assume that $\sqrt{\ell}\cdot T(n)\gg \ell n^2$, and hence, the runtime of $\BBB$ is at most $\sqrt{\ell} \cdot T(n)$. Note that this algorithm $\BBB$ takes a key $k$ and boxes $c_1,\dots,c_n$, and returns a transcript $(q_1,z_1,\dots,q_{\ell},z_{\ell})$. Now observe that the event $\left\{\exists i\in[\ell] \text{ s.t.\ }
\left| q_i(\mu)-q_i(\vec{x}) \right|\geq0.39\right\}$ can be easily verified by a distinguisher $\DDD$ (that knows the distribution $\mu$ and the plaintext inputs $\vec{x}$).\footnote{As the distribution $\mu$ is uniform on $O(n)$ points, and as $\DDD$ knows both $\mu$ and $\vec{x}$, it can compute $q_1(\mu),q_1(\vec{x}),\dots,q_\ell(\mu),q_\ell(\vec{x})$ in time $O(\ell\cdot n\cdot n^2)=O(n^5)$. By setting $T(n)=\Omega(n^5)$ we make sure that this computation in within the distinguisher's limits. We remark that it is possible to speedup the distinguisher using sampling, but this is not required with our choice of parameters.} Thus, by Definition~\ref{def:boxes}, there exists a simulator $\SSS$ such that
$$
\Pr_{
\substack{\vec{x}\sim\mu^n\\ 
(q_1,z_1,\dots,q_\ell,z_\ell)\leftarrow
\SSS^{\Oracle_{2b}(\vec{x})}
}}\Big[
\exists i\in[\ell] \text{ s.t.\ }
\left| q_i(\mu)-q_i(\vec{x}) \right|\geq0.39
\Big]\geq\frac{1}{4}-\frac{1}{10}=0.15.
$$
This is a contradiction because for any algorithm $\SSS$ that  computes a function $q$ based on a small fraction of its input sample $\vec{x}\sim\mu^n$,  we have that w.h.p.\ $q(\vec{x})\approx q(\mu)$. Specifically, as the points $x_1,\dots,x_n$ are sampled independently from $\mu$, for every fixture of the queries $(q_1,\dots,q_\ell)$, by the Hoeffding bound (over the $n-2b$ points not observed by $\SSS$), for
every $\delta > 0$ and $\epsilon >0$, there exists a large enough $n$, such that with probability $1-\delta$, for every $i$, it holds that
$$
|q_i(\vec{x})-q_i(\mu)|\leq \frac{2b}{n}+\epsilon\ll 0.39.
$$
The last inequality holds if we pick $\epsilon$ small enough.
\end{proof}

\bibliographystyle{alpha}
\newcommand{\etalchar}[1]{$^{#1}$}

\appendix

\section{Proof of \Cref{lem:sparsify from expander}}

\label{sec:proof sparsify}

Below, we prove \Cref{lem:sparsify from expander}. First, we recall
that a $(1+\epsilon)$-spectral sparsifier of an expander can be obtained
via a very simple algorithm.
\begin{lem}
[Corrollary 11.4 of \cite{bernstein2020fully}]\label{lem:sampling expander}Let
$\epsilon\in(0,1/2)$ and $G=(V,E)$ be a $\phi$-expander. By independently
sampling every edge $\{u,v\}\in E$ with probability
\[
p_{uv}\ge\min\left\{ 1,\left(\frac{24\log n}{\epsilon\phi^{2}}\right)^{2}\cdot\frac{1}{\min\{\deg(u),\deg(v)\}}\right\} 
\]
and setting the weights of a sampled edge as $1/p_{uv}$, then with
probability at least $1-1/n^{3}$, the resulting graph $\tilde{G}$
is a $(1+\epsilon)$-spectral sparsifier of $G$. 
\end{lem}

As the above lemma requires independent sampling, but the naive implementation
of \Cref{lem:sampling expander} would takes $O(|E(G)|)$ time which
is too slow for us. We need the following tool for efficiency:
\begin{lem}
[Efficient Subset Sampling \cite{bringmann2012efficient}]\label{lem:fast sampling}Given
a universe $U$ of size $n$ and a sampling probability $p$. Then,
we can compute a set $S\subseteq U$ in worst-case time $O(pn\log n)$
where each element of $U$ is in $S$ with probability $p$ independently.
The algorithm succeeds with high probability.
\end{lem}

With the above tools, we are now ready to prove \Cref{lem:sparsify from expander}.
\begin{proof}
[Proof of \Cref{lem:sparsify from expander}]To implement $\sparsify(t)$,
we will need some data structure on each expander $G_{i}$. For each
node $u\in V(G_{i})$, we maintain an approximate degree $\widetilde{\deg}_{G_{i}}(u)$
which is $\deg_{G_{i}}(u)$ rounded up to the closest power of $2$.
For each edge $\{u,v\}\in E(G_{i})$, we say that $u$ \emph{owns}
$\{u,v\}$ if $\widetilde{\deg}_{G_{i}}(u)\le\widetilde{\deg}_{G_{i}}(v)$.
For each node $u$, we maintain a collection of edges owned by $u$
in a binary search tree. In this way, given $i$, we can obtain the
$i$-th edge owned by $u$ in $O(\log n)$ time, which is needed by
the subset sampling algorithm from \Cref{lem:fast sampling}. Whenever
$\widetilde{\deg}_{G_{i}}(u)$ decreases, we update $O(\widetilde{\deg}_{G_{i}}(u))$
edges over all binary search trees, but $\widetilde{\deg}_{G_{i}}(u)$
changes only after $\Omega(\widetilde{\deg}_{G_{i}}(u))$ edge deletions
to $G_{i}$. So the amortized cost for maintaining lists that each
vertex owns is $O(\log n)$. We will use the following claim.
\begin{claim}
\label{claim:sparsify expander}We can compute $(1+\epsilon)$-spectral-sparsifier
$H_{i}$ of $G_{i}$ in $\Otil(|V(G_{i})|)$ time. Moreover, given
$t_{i}$ updates to $H_{i}$ from an oblivious adversary, the total
update time and the total number of edge changes in $H_{i}$ is at
most $\Otil(t_{i})$. 
\end{claim}

\begin{proof}
To compute a $(1+\epsilon)$-spectral-sparsifier $H_{i}$ in $\Otil(|V(G_{i})|)$
time, we do the following. For each node $u$ and for each edge $\{u,v\}$
owned by $u$, we independently sample $\{u,v\}$ (and then scale
appropriately) with probability
\begin{align*}
\tilde{p}_{u}:=\min\left\{ 1,\left(\frac{24\log n}{\epsilon\phi^{2}}\right)^{2}\cdot\frac{2}{\widetilde{\deg}_{G_{i}}(u)}\right\}  & =\min\left\{ 1,\left(\frac{24\log n}{\epsilon\phi^{2}}\right)^{2}\cdot\frac{2}{\min\{\widetilde{\deg}_{G_{i}}(u),\widetilde{\deg}_{G_{i}}(v)\}}\right\} \\
 & \ge\min\left\{ 1,\left(\frac{24\log n}{\epsilon\phi^{2}}\right)^{2}\cdot\frac{1}{\min\{\deg_{G_{i}}(u),\deg_{G_{i}}(v)\}}\right\} 
\end{align*}
where the last inequality is because $\deg_{G_{i}}(u)\le\widetilde{\deg}_{G_{i}}(u)\le2\deg_{G_{i}}(u)$.
By \Cref{lem:sampling expander}, we have that $H_{i}$ is a $(1+\epsilon)$-spectral
sparsifier of $G_{i}$. By \Cref{lem:fast sampling}, we can sample
all edges owned by $u$ in time $\Otil(\tilde{p}_{u}\deg_{G_{i}}(u))=\Otil(1)$.
So the total time is $\Otil(|V(G_{i})|)$. Whenever $\widetilde{\deg}_{G_{i}}(u)$
changes we just resample edges owned by $u$, which takes $\Otil(1)$
time. Given $t_{i}$ updates to $G_{i}$, the total update time and
the total number of edge changes in $H_{i}$ is at most $\Otil(t_{i})$.
If the updates are independent from $H_{i}$, then $H_{i}$ will remain
$(1+\epsilon)$-spectral sparsifier of $G_{i}$ by \Cref{lem:sampling expander}.
\end{proof}
Now, we are ready to implement $\sparsify(t)$. We take each expander
$G_{i}$ and compute a $(1+\epsilon)$-spectral sparsifier $H_{i}$
of $G_{i}$ using \Cref{claim:sparsify expander} in $\Otil(|V(G_{i})|)$
time. We maintain $H=\cup_{i}H_{i}$ as the sparsifier returned by
$\sparsify(t)$. In total this takes $\sum_{i=1}^{z}\Otil(|V(G_{i})|)=\Otil(n)$
time by \Cref{thm:dynamic exp decomp}. Next, each edge update to $G$
generates some edge updates to $G_{i}$. For each update to $G_{i}$,
we maintain $H_{i}$ using \Cref{claim:sparsify expander}. As long
as there are less than $t$ updates to all $G_{i}$ in total (i.e.~$\counter$
has not increased by more than $t$), the total time to maintain $H_{i}$
is $\Otil(t)$ and the total number of edge changes in $H$ is also
$\Otil(t)$. Therefore, until $\counter$ has increased by $t$, $H$
contains at most $\Otil(n+t)$ edges and the total time to maintain
$H$ is $\Otil(n+t)$ as well. As \Cref{lem:sampling expander} guarantees
correctness against an oblivious adversary, if the edge updates before
$\counter$ has increased by more than $t$ are independent from $H$,
then with high probability, $H$ is a $(1+\epsilon)$-spectral sparsifier
throughout the updates.

This complete the proof of \Cref{lem:sparsify from expander}.
\end{proof}

\end{document}